\documentclass{article}

\usepackage{amsmath}
\usepackage{algorithm}
\usepackage{algcompatible}
\usepackage{graphicx,psfrag,amsmath,amssymb,epsf}
\usepackage{epsfig,graphics}

\newtheorem{theorem}{Theorem}[section]
\newtheorem{lemma}{Lemma}[section]
\newtheorem{corollary}{Corollary}[section]
\newtheorem{definition}{Definition}[section]
\newtheorem{proposition}{Proposition}[section]

\newcommand{\qed}{\hfill\hbox{\rlap{$\sqcap$}$\sqcup$}}
\newenvironment{proof}{\noindent \emph{Proof.\,}}{\qed}

\algnewcommand\algorithmicreturn{\textbf{return}}
\algnewcommand\RETURN{\algorithmicreturn}
\algnewcommand\algorithmicprocedure{\textbf{procedure}}
\algnewcommand\PROCEDURE{\item[\algorithmicprocedure]}%
\algnewcommand\algorithmicendprocedure{\textbf{end procedure}}
\algnewcommand\ENDPROCEDURE{\item[\algorithmicendprocedure]}%
\algnewcommand{\algvar}[1]{{\text{\ttfamily\detokenize{#1}}}}
\algnewcommand{\algarg}[1]{{\text{\ttfamily\itshape\detokenize{#1}}}}
\algnewcommand{\algproc}[1]{{\text{\ttfamily\detokenize{#1}}}}
\algnewcommand{\algassign}{\leftarrow}

\title{Men Can't Always be Transformed into Mice:  Decision Algorithms and Complexity for Sorting by Symmetric Reversals}

\author{Xin Tong\footnote{College of Computer Science and Technology, Shandong University, Qingdao, China. Email: {\tt xtong@mail.sdu.edu.cn}.}
\and
Yixiao Yu\footnote{College of Computer Science and Technology, Shandong University, Qingdao, China. Email: {\tt yixiaoyu@mail.sdu.edu.cn}.}
\and
Ziyi Fang\footnote{College of Computer Science and Technology, Shandong University, Qingdao, China. Email: {\tt fangziyi@mail.sdu.edu.cn}.}
\and
Haitao Jiang\footnote{College of Computer Science and Technology, Shandong University, Qingdao, China. Email: {\tt htjiang@sdu.edu.cn}.}
\and
Lusheng Wang\footnote{Department of Computer Science, City University of Hong Kong, Kowloon, Hong Kong, China. Email: {\tt cswangl@cityu.edu.hk}.}
\and
Binhai Zhu\footnote{Gianforte School of Computing, Montana State University, Bozeman, MT 59717, USA. Email: {\tt bhz@montana.edu}.}
\and
Daming Zhu\footnote{College of Computer Science and Technology, Shandong University, Qingdao, China. Email: {\tt dmzhu@sdu.edu.cn}.}
}

\date{}

\begin{document}

\maketitle

\begin{abstract}
Sorting a permutation by reversals is a famous problem in genome rearrangements, and has been well studied over the past thirty years.
But the involvement of repeated segments is inevitable during genome evolution, especially in reversal events.
Since 1997, quite some biological evidence were found that in many genomes the reversed regions are usually flanked by a pair of inverted repeats. 
For example, a reversal will transform $+a +x -y -z -a$ into $+a +z +y -x -a$, where $+a$ and $-a$ form a pair of inverted repeats.

This type of reversals are called symmetric reversals, which, unfortunately, were largely ignored until recently.
In this paper, we investigate the problem of sorting by symmetric reversals, which requires a series of symmetric reversals to transform one chromosome $A$ into the another chromosome $B$. The decision problem of sorting by symmetric reversals is referred to as {\em SSR} (when the input chromosomes $A$ and $B$ are given, we use {\em SSR(A,B)}, similarly for the following optimization version) and the corresponding optimization version
(i.e., when the answer for {\em SSR(A,B)} is yes, using the minimum number of symmetric reversals to convert $A$ to $B$), is referred to as {\em SMSR(A,B)}. The main results of this paper are summarized as follows, where the input is a pair of chromosomes $A$ and $B$ with $n$ repeats.
\begin{enumerate}
\item We present an $O(n^2)$ time algorithm to solve the decision problem {\em SSR(A,B)}, i.e., determine whether a chromosome $A$ can be transformed into $B$ by a series of symmetric reversals. This result is achieved by converting the problem to the circle graph, which has been augmented significantly from the traditional circle graph and a list of combinatorial properties must be proved to successfully answer the decision question. 

\item We design an $O(n^2)$ time algorithm for a special 2-balanced case of {\em SMSR(A,B)}, where chromosomes $A$ and $B$ both have duplication number 2 and every repeat appears twice in different orientations in $A$ and $B$.

\item We show that SMSR is NP-hard even if the duplication number of the input chromosomes are at most 2, hence showing that the above positive optimization result is the best possible. As a by-product, we show that the \emph{minimum Steiner tree} problem on \emph{circle graphs} is NP-hard, settling the complexity status of a 38-year old open problem.
\end{enumerate}

\end{abstract}

\section{Introduction}
In the 1980s, quite some evidence was found that some species have essentially the same set of genes, but their gene order differs~\cite{HP,PH}.
Since then, sorting permutations with rearrangement operations has gained a lot of interest in the area of computational biology in the last thirty years.
Sankoff {\em et al.} formally defined the genome rearrangement events with some basic operations on genomes, e.g., reversals, transpositions and translocations~\cite{SLAP},
where the reversal operation is adopted the most frequently~\cite{KS,GF,WPRC}.

The complexity of the problem of sorting permutations by reversals is closely related to whether the genes are signed or not.
Watterson {\em et al.} pioneered the research on sorting an unsigned permutation by reversals~\cite{WEHM}.
In 1997, Caprara established the NP-hardness of this problem \cite{Caprara}. Soon after, Berman {\em et al.} showed it to be APX-hard~\cite{PBMK}.
Kececioglu and Sankoff presented the first polynomial time approximation for this problem with a factor of 2~\cite{KS}.
The approximation ratio was improved to 1.5 by Christie~\cite{Christie}. So far as we know, the best approximation ratio for the problem of sorting an unsigned permutation by reversals is 1.375~\cite{BHK}.
As for the more realistic problem of sorting signed permutations by reversals, Hannenhalli and Pevzner proposed an $O(n^4)$ time exact algorithm for this problem, where $n$ is the number of
genes in the given permutation (genomes) \cite{SHPP}. The time complexity was later improved to $O(n^2)$ by Kaplan {\em et al.}~\cite{KST}. The current best running time is $O(n^{1.5}\sqrt{\log n})$ by Tannier {\em et al.}~\cite{TBS}.

On the other hand, some evidence has been found that the breakpoints where reversals occur could have some special property in the genomes~\cite{Longo,Sankoff}.
As early as in 1997, some studies showed that the breakpoints are often associated with repetitive elements on mammals and drosophila genomes~\cite{Thomas,Bailey,Armengol,Small}.
In fact, the well-known ``site-specific recombination'',  which has an important application in ``gene knock out"~\cite{Sauer,SauH,OCM}, also fulfills this rule.
However, it was still not clear why and how repetitive elements play important roles in genome rearrangement.
Recently, Wang {\em et al.} conducted a systematic study on comparing different strains of various bacteria
such as Pseudomonas aeruginosa, Escherichia coli, Mycobacterium tuberculosis and Shewanella \cite{Dan1,Dan2}. Their study further illustrated  that repeats are associated with the ends of rearrangement segments for various rearrangement events such as reversal, transposition, inverted block interchange, etc, so that the left and right neighborhoods of those repeats remain unchanged  after the rearrangement events.
Focusing on reversal events, the reversed regions are usually flanked by a pair of inverted repeats \cite{Small}. 
The following real example is from Pseudomonas aeruginosa strains in \cite{Dan2}. Such a phenomenon can also better explain why the famous ``breakpoint reuse'' (which were an interesting finding and discussed in details when comparing human with mouse) happen \cite{PT}.
\begin{figure}[htbp]
  \centering
  \includegraphics[width=0.6\textwidth]{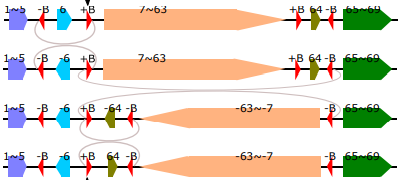}
  \caption{Three symmetric reversals use the repeat `$+B$' three times.}
  \label{consistentornot}
\end{figure} 

In this paper,
we propose a new  model called {\it sorting by symmetric reversals}, which requires each inverted region on the chromosomes  being flanked by a pair of mutually inverted repeats.
We investigate the decision problem of sorting by symmetric reversals (SSR for short), which asks whether a chromosome can be transformed into the other  by a series of symmetric reversals.
We devise an $O(n^2)$ time algorithm to solve this decision problem.
We also study the optimization version (referred to as SMSR) that uses a minimum number of symmetric reversals to transform one chromosome into the other. We design an $O(n^2)$ time algorithm for a special 2-balanced case of SMSR,  where chromosomes have duplication number 2 and every repeat appears twice in different orientations in each chromosome.
 We finally show that the optimization problem SMSR is NP-hard even if each repeat has at most 2 duplications in each of the input chromosome.

In the NP-hardness proof, we set up the relationship between our problem and the \emph{minimum Steiner tree} problem on \emph{circle graphs}.
The minimum Steiner tree problem on circle graphs has been considered to be in $P$ as indicated by Johnson in 1985 \cite{Johnson}. Recently,
Figueiredo {\em et al.} revisited Johnson's table and  still marked the problem   as in  $P$ \cite{FMSS}, while leaving the  reference as ``ongoing''.
Here we clarify that the \emph{minimum Steiner tree} problem on \emph{circle graphs} is in fact NP-hard, settling this 38-year old open problem.

This paper is organized as follows. In Section 2, we give some  definitions.  We then present an algorithm to  solve SSR  under a special case,  where the duplication number of the input chromosomes is 2 in Section 3.
In Section 4, we present a polynomial algorithm for SMSR for the special 2-balanced case.
In Section 5, we present an algorithm to solve SSR for the general case. In Section 6, we show that SMSR is NP-hard for the case that chromosomes have duplication number 2, with the help of the new NP-hardness result on the minimum Steiner tree problem on circle graphs. 
Finally, conclusions are  given in Section 7.

\section{Preliminaries}
In the literature of genome rearrangement, we always have a set of integers $\Sigma_1=\{1, \cdots, g\}$, where  each integer stands for a long DNA sequence (syntenic block or a gene).
For  simplicity, we use ``gene'' hereafter.
Since we will study symmetric reversals, we define $\Sigma_2=\{r_0, r_1, r_2, \cdots, r_t\}$ to be a set of symbols, each of them is referred to as a {\it repeat} and represents a relative shorter DNA sequence compared with genes. We then set $\Sigma =\Sigma_1\cup \Sigma _2$ to be the alphabet for the whole chromosome.

Since reversal operations work on a chromosome internally, a genome can be considered as a chromosome for our purpose, i.e., each genome is a singleton and contains only one chromosome.
Here we assume that each gene appears exactly once on a chromosome, on the other hand, by name, a repeat could appear multiple times.
A gene/repeat $x$ on a chromosome may appear in two different orientations, i.e., either as $+x$ or $-x$.
Thus, each chromosome of interest is presented by a sequence of signed integers/symbols.

The number of occurrences of a gene/repeat $x$ in both orientations is called the {\em duplication number} of $x$ on the chromosome $\pi$, denoted by $dp[x,\pi]$.
The duplication number of a chromosome $\pi$, denoted by $dp[\pi]$, is the maximum duplication number of the repeats on it.
For example, chromosome $\pi=[+r_0,+1,-r,+2,$ $+r,-r_0]$, $dp[1,\pi]=dp[2,\pi]=1$, $dp[r_0,\pi]=dp[r,\pi]=2$, and $dp[\pi]=2$.
Two chromosomes $\pi_{1}$ and $\pi_{2}$ are \emph{related} if their duplication numbers for all genes and repeats  are identical.
Let $|x|\in \Sigma$ be an integer or symbol, and $+|x|$ and $-|x|$ be two occurrences of $|x|$, where the orientations of $+|x|$ and $-|x|$ are different. A chromosome of $n$ genes/repeats is denoted as $\pi= [x_{1}, x_{2}, \dots, x_{n-1}, x_{n}]$.
A linear chromosome has two ends, and it can be read from either end to the other,
so the chromosome $\pi= [x_{1}, x_{2}, \dots, x_{n-1}, x_{n}]$ can also be described as $[-x_{n}, -x_{n-1}, \dots, -x_{2}, -x_{1}]$,
which is called the {\em reversed and negated} form of $\pi$.

A \emph{reversal} is an operation that reverses a segment of continuous integers (or symbols) on the chromosome.
A \emph{symmetric reversal} is a reversal, where   the reversed segment is flanked by pair of identical repeats with different orientations, i.e, either $(+r, \cdots, -r)$ or
$(-r, \cdots, +r)$ for some $r\in \Sigma_2$.
In other words, let  $\pi= [x_{1}, x_{2}, \dots, x_{n}]$ be a chromosome.  The reversal $\rho(i, j)$ ($1\leq i<j\leq n$) reverses the segment $[x_{i}, x_{i+1}, \dots, x_{j-1}, x_{j}]$, and yields $\pi'=[x_{1}, x_{2},\dots, x_{i-1}, -x_{j}, -x_{j-1}$, $\dots, -x_{i+1}, -x_{i}, x_{j+1},$ $\dots, x_{n}]$.
If $x_{i}=-x_{j}$, we say that $\rho(i, j)$ is a symmetric reversal on $|x_{i}|$.
Reversing a whole chromosome will not change the relative order of the integers but their signs, so we assume that each chromosome is flanked by
$+r_0$ and $-r_0$, then a chromosome will turn into its reversed and negated form by performing a symmetric reversal between $+r_0$  and $-r_0$.

Again, as a simple example, let $\pi=[+r_0, +1, -r_{1}, +2, +r_{2}, +r_{1}, +r_{2}, -r_0]$, then a symmetric reversal on $r_{1}$ yields
$\pi'=[+r_0, +1, -r_{1}, -r_{2}, -2, +r_{1}, +r_{2}, -r_0]$.

Now, we formally define the problems to be investigated in this paper.
\begin{definition}{Sorting by Symmetric Reversals,  \textbf{SSR} for short.}

\textbf{Instance:}~~Two related chromosomes $\pi$ and $\tau$, such that $dp[\pi]=dp[\tau]\geq2$.

\textbf{Question:}~~Is there a sequence of symmetric reversals that transform $\pi$ into $\tau$?.

\end{definition}

\begin{definition}{Sorting by the Minimum Symmetric Reversals, \textbf{SMSR} for short.}

\textbf{Instance:}~~Two related chromosomes $\pi$ and $\tau$ with $dp[\pi]=dp[\tau]\geq2$, and an integer $m$.

\textbf{Question:}~~Is there a  sequence of symmetric reversals $\rho_{1}, \rho_{2}, \dots, \rho_{m}$ that transform $\pi$ into $\tau$, such that $m$ is minimized?

\end{definition}

There is a standard way to make a signed gene/repeat unsigned. Let $\pi= [x_{0}, x_{1}, \dots, x_{n+1}]$ be a chromosome, each occurrence of gene/repeat of $\pi$, say $x_{i}$ ($0\leq i\leq n+1$), is represented by a pair of ordered nodes, $l(x_{i})$ and $r(x_{i})$. If the sign of $x_{i}$ is +, then $l(x_{i})=|x_{i}|^{h}$ and $r(x_{i})=|x_{i}|^{t}$; otherwise, $l(x_{i})=|x_{i}|^{t}$ and $r(x_{i})=|x_{i}|^{h}$. Note that, if $x_i$ and $x_j$ ($i\neq j$) are different occurrences of the same repeat, i.e., $|x_i|=|x_j|$,  $l(x_{i})$, $l(x_{j})$, $r(x_i)$ and $r(x_j)$   correspond to two nodes $|x_{i}|^{h}$ and $|x_{i}|^{t}$ only.
Consequently, $\pi$ will also be described as $[l(x_{0}), r(x_{0}), l(x_{1}), r(x_{1}), \dots, l(x_{n+1}), r(x_{n+1})]$.
We say that $r(x_{i})$ and $l(x_{i+1})$, for $0\leq i\leq n$, form an \emph{adjacency}, denoted by $\langle r(x_{i}), l(x_{i+1})\rangle$. (Note that in the signed representation
of a chromosome $\pi$, we simply say that $\langle x_i,x_{i+1}\rangle$ forms an adjacency; moreover, $\langle x_i,x_{i+1}\rangle=\langle -x_{i+1},-x_i\rangle$.) Also, we say that the adjacency $\langle r(x_{i}), l(x_{i+1})\rangle$ is associated with $x_{i}$ and $x_{i+1}$.
Let $\mathcal{A}[\pi]$ represent the multi-set of adjacencies of $\pi$.
We take the chromosome $\pi=[+r_0, +1, -r_1, +2, +r_1, -r_0]$ as an example to explain the above notations. The multi-set of adjacencies is $\mathcal{A}[\pi]=\{\langle r_{0}^{t}, 1^{h}\rangle$, $\langle1^{t}, r_{1}^{t}\rangle$, $\langle r_{1}^{h}, 2^{h}\rangle$, $\langle2^{t}, r_{1}^{h}\rangle$, $\langle r_{1}^{t}, r_{0}^{t}\rangle\}$,
$\pi$ can also be viewed as $[r_{0}^{h}, r_{0}^{t}, 1^{h}, 1^{t}, r_{1}^{t}, r_{1}^{h}, 2^{h}, 2^{t}, r_{1}^{h}, r_{1}^{t}, r_{0}^{t}, r_{0}^{h}]$.
\begin{lemma}
\label{preserve}
Let  $\pi$ be a chromosome and $\pi'$ is obtained from $\pi$ by performing a symmetric reversal.
Then
$\mathcal{A}[\pi]=\mathcal{A}[\pi']$.
\end{lemma}

\begin{proof}
It is apparent that performing the symmetric reversal between $+r_0$ and $-r_0$ will not change $\mathcal{A}[\pi]$.
Assume that the symmetric reversal $\rho(i, j)$ is performed on the chromosome $\pi= [x_{0}, x_{1}, \dots, x_{n+1}]$, such that $x_{i}=-x_{j}$, where $1\leq i<j\leq n$,
and yields $\pi'=\pi\bullet\rho(i, j)=[x_{0}, x_{1},\dots, x_{i-1}, -x_{j}, -x_{j-1}, \dots$, $-x_{i+1}, -x_{i}, x_{j+1}, \dots, x_{n+1}]$.
Then $\rho(i, j)$ breaks two adjacencies $\langle r(x_{i-1}), l(x_{i})\rangle$ and $\langle r(x_{j}), l(x_{j+1})\rangle$, and creates two
new adjacencies $\langle r(x_{i-1}), l(-x_{j})\rangle$ and $\langle r(-x_{i}), l(x_{j+1})\rangle$.
Since $x_{i}=-x_{j}$, $l(x_{i})=l(-x_{j})$ and $r(x_{j})=r(-x_{i})$, thus
$\langle r(x_{i-1}), l(x_{i})\rangle = \langle r(x_{i-1}), l(-x_{j})\rangle$ and $\langle r(x_{j}), l(x_{j+1})\rangle = \langle r(-x_{i}), l(x_{j+1})\rangle$.
Consequently, $\mathcal{A}[\pi]=\mathcal{A}[\pi']$.
\end{proof}

Actually, Lemma~\ref{preserve} implies a necessary condition for answering the decision question of \textbf{SSR}.
\begin{theorem}
\label{necessary}
Chromosome $\pi$ cannot be transformed into $\tau$ by a series of symmetric reversals if $\mathcal{A}[\pi]\neq\mathcal{A}[\tau]$.
\end{theorem}

A simple negative example would be $\pi=[+r_0$, $+r_1$, $-2$, $+r_1$, $-1$, $-r_0]$ and
$\tau=[+r_0$,$-r_1$,$+2$,$-r_1,$ $+1$,$-r_0]$.
One can easily check that $\mathcal{A}[\pi]\neq\mathcal{A}[\tau]$, which means that there is no way to convert $\pi$ to $\tau$ using symmetric reversals.
In the next section, as a warm-up, we first solve the case when each
repeat appears at most twice in $\pi$ and $\tau$. Even though the method is not
extremely hard, we hope the presentation and some of the concepts can help readers understand the details for the general case in Section 5 better.

\section{An $O(n^2)$ Algorithm for SSR with Duplication Number 2}
In this section, we consider a special case, where the duplication numbers for
the two related chromosomes $\pi$ and $\tau$ are both $2$. That is, $\mathcal{A}[\pi]=\mathcal{A}[\tau]$ and $dp[\pi]=dp[\tau]=2$.
We will design an algorithm with running time $O(n^2)$ to determine if there is a sequence of symmetric reversals that transform $\pi$ into $\tau$.

Note that $\mathcal{A}[\pi]$ is a multi-set, where an adjacency  may appear more than once. 
When the duplication number of each repeat in the chromosome is at most 2, the same adjacency can appear at most twice in $\mathcal{A}[\pi]$.

Let $\pi= [x_{0}, x_{1}, \dots, x_{n+1}]$ be a chromosome.
Let 
$x_i$ and $x_j$ be the two occurrences of a repeat $x$, 
and $x_{i+1}$ and $x_{j+1}$ the two occurrences of the other repeat 
$x'$ in $\pi$.
We say that $|x_i|$ and $|x_{i+1}|$ are {\it redundant}, if 
 $r(x_{i})=r(x_{j})$ and $l(x_{i+1})=l(x_{j+1})$ (or $r(x_{i})=l(x_{j})$ and $l(x_{i+1})=r(x_{j-1})$). 
In this case, the adjacency $\langle r(x_i), l(x_{i+1} \rangle $  appears twice. In fact, it is the only case that an adjacency can appear twice.
 An example is as follows: $\pi=[+r_0, +r_1, -r_2, +1, +r_2, -r_1,-r_0]$, where the adjacency $\langle +r_1,-r_2\rangle$ appears twice (the second negatively), hence $r_1$ and $r_2$ are redundant.
 The following lemma tells us that if $x_i$ and $x_j$ are redundant, 
 we only need to use one of them to do reversals and the other can be deleted from the chromosome so that each adjacency appears only once.
 
 \begin{lemma}
\label{norepetitive}
Given two chromosomes  $\pi= [x_{0}, x_{1}, \dots, x_{n+1}]$ and $\tau$, such that $\mathcal{A}[\pi]=\mathcal{A}[\tau]$.
Let $|x_i|$ and $|x_{i+1}|$ be two    repeats in $\pi$ that are redundant.
Let $\pi'$ and $\tau'$ be the chromosomes after deleting the two occurrences of $|x_{i+1}|$ from both $\pi$ and $\tau$, respectively.
Then $\pi$ can be transformed into $\tau$ by a series of symmetric reversals if and only if $\pi'$ can be transformed into $\tau'$ by a series of symmetric reversals.
\end{lemma}

\begin{proof}
Without loss of generality, we assume that $r(x_{i})=r(x_{j})=|x_{i}|^{a}$ and $l(x_{i+1})=l(x_{j+1})=|x_{i+1}|^{b}$, where $a, b\in \{h, t\}$. The proof of the other case is similar.

$(\Rightarrow)$ Assume that there is a series of symmetric reversals, $\rho_{1}, \rho_{2}, \dots, \rho_{m}$, that transforms $\pi$ into $\tau$,
to be specific, $\pi_{0}=\pi$, $\pi_{k}=\pi_{k-1}\bullet\rho_{k}$ for each $1\leq k\leq m$, and $\pi_{m}=\tau$.
Suppose that there exists a symmetric reversal $\rho_{k}$($1\leq k\leq m$), which is on $|x_{i+1}|$.
Lemma~\ref{preserve} guarantees that $\langle |x_{i}|^{a}, |x_{i+1}|^{b}\rangle$ still appears twice in $\mathcal{A}[\pi_{k-1}]$.
Because the two occurrences of $|x_{i+1}|$ have distinct sign in $\pi_{k-1}$, the two $|x_{i}|^{a}$s are located at different sides of the two $|x_{i+1}|^{b}$s respectively,
which implies that some $|x_{i}|^{a}$ is the left node of some occurrence of $|x_{i}|$ and the other $|x_{i}|^{a}$ is the right node of the other occurrence of $|x_{i}|$.
Therefore the signs of the two occurrences of $|x_{i}|$ are also distinct in $\pi_{k-1}$.
It is apparent that performing the symmetric reversal on $|x_{i}|$ would also transform $\pi_{k-1}$ into $\pi_{k}$.

$(\Leftarrow)$ Assume that there is a series of symmetric reversals, $\rho'_{1}, \rho'_{2}, \dots, \rho'_{m'}$ which transforms $\pi'$ into $\tau'$.
Let the corresponding chromosomes be $\pi'_{0}=\pi'$, $\pi'_{1}$, $\dots$, $\pi'_{m'}=\tau'$, where $\pi_{k'}=\pi_{k'-1}\bullet\rho'_{k'}$ for each $1\leq k'\leq m'$.
We can obtain $\overline{\pi'}_{k'}$ by substituting $x_{i}$ with $[x_{i}, x_{i+1}]$ and $-x_{i}$ with $[-x_{i+1}, -x_{i}]$ in $\pi'_{k'}$.
Clearly, $\overline{\pi'}_{0}=\pi$ and $\overline{\pi'}_{m'}=\tau$.
For each $1\leq k'\leq m'$, $\rho'_{k'}$ is applicable to $\overline{\pi'}_{k'-1}$, since all the elements on $\overline{\pi'}_{k'-1}$ have the same signs with those on $\pi'_{k'-1}$;
also, $\overline{\pi'}_{k'}=\overline{\pi'}_{k'-1}\bullet \rho'_{k'}$, since the two adjacencies of the form $\langle |x_{i}|^{a}, |x_{i+1}|^{b}\rangle$ can not be changed by $\rho'_{k'}$.
\end{proof}

Regarding the previous example, $\pi=[+r_0, +r_1, -r_2, +1, +r_2, -r_1,-r_0]$, where $r_1$ and $r_2$ are redundant, following the above lemma, one can obtain $\pi'=[+r_0, +r_1, +1, -r_1,-r_0]$. 
This is in fact equivalent to replacing the adjacency
$\langle +r_1,-r_2\rangle$ by $r_1$, and $\langle +r_2,-r_1\rangle$ by $-r_1$.

A chromosome $\pi$ is \emph{simple} if every adjacency in   $\mathcal{A}[\pi]$ appears only once.
Based on Lemma~\ref{norepetitive}, we can remove the two occurrences of a redundant repeat  from the chromosomes.
Thus, 
if $dp[\pi]=dp[\tau]=2$, we can always assume that both $\pi$ and $\tau$ are simple.
Consequently, there is a unique bijection between two corresponding adjacency sets $\mathcal{A}[\pi]$ and $\mathcal{A}[\tau]$. We say that any pair of identical adjacencies are matched to each other.

For each repeat $x$ with $dp[\pi, x]=dp[\tau, x]=2$, let $x_{i}$, $x_{j}$ be the two occurrences of $x$ in $\pi$, and $y_{i'}$, $y_{j'}$ be the two occurrences of $x$ in $\tau$,
there are four adjacencies associated with $x_{i}$ and $x_{j}$ in $\pi$: $\langle r(x_{i-1}), l(x_{i})\rangle$, $\langle r(x_{i}), l(x_{i+1})\rangle$, $\langle r(x_{j-1}), l(x_{j})\rangle$, $\langle r(x_{j}), l(x_{j+1})\rangle$. Similarly, there are four adjacencies associated with $y_{i'}$ and $y_{j'}$ in $\tau$.
We say that $x$ is an \emph{neighbor-consistent} repeat, if $\langle r(x_{i-1}), l(x_{i})\rangle$ and $\langle r(x_{i}), l(x_{i+1})\rangle$ are matched to two adjacencies both associated with $y_{i'}$ or both associated with $y_{j'}$. That is, the left and right neighbors of $x
_i$ are identical in both chromosomes.
Note that $\mathcal{A}[\pi]=\mathcal{A}[\tau]$ also implies that  the left and right neighbors of the other  occurrences $x_j$ are also identical in both two chromosomes if $x$ is neighbor-consistent.
If $\langle r(x_{i-1}), l(x_{i})\rangle$ and $\langle r(x_{i}), l(x_{i+1})\rangle$ are matched to two adjacencies, one of which is associated with $y_{i'}$ and the other is associated with $y_{j'}$, 
then $x$ is an \emph{neighbor-inconsistent} repeat. The genes and the repeats which appear once in $\pi$ are also defined to be neighbor-consistent. (See Figure~\ref{consistentornot} for an example.)
By definition and the fact that $\mathcal{A}[\pi]=\mathcal{A}[\tau]$, we have 
\begin{proposition}
\label{w1}
Performing a symmetric reversal on a repeat will turn the repeat  from neighbor-consistent to neighbor-inconsistent or vice versa. (See Figure~\ref{consistentornot}.)
\end{proposition}

\begin{figure}[htbp]
  \centering
  \includegraphics[width=1.0\textwidth]{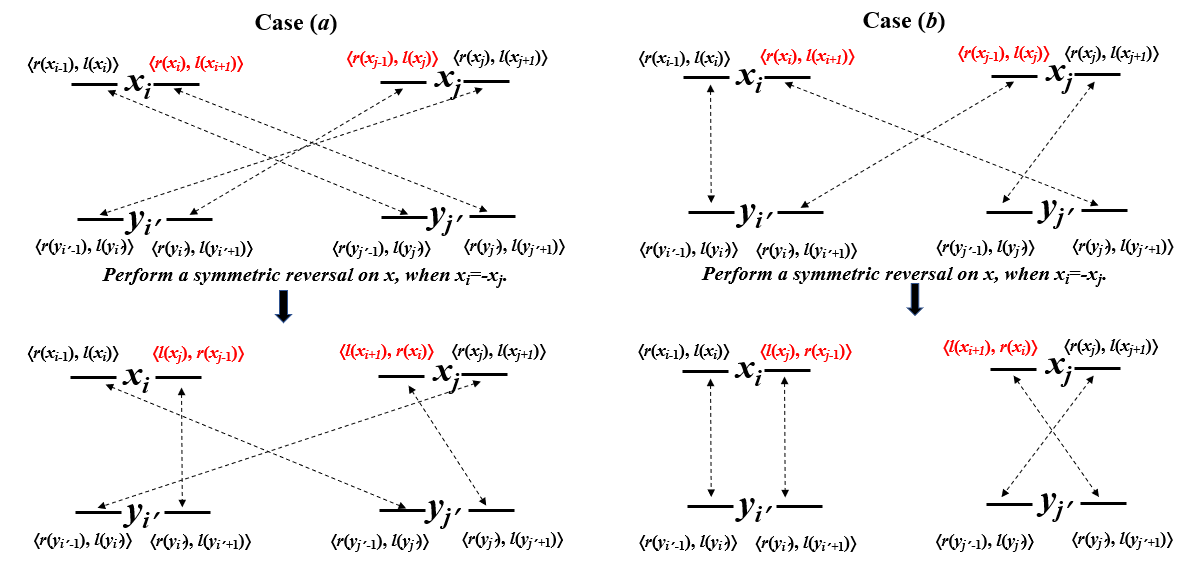}
  \caption{$x_{i}$, $x_{j}$ are the two occurrences of $x$ in $\pi$ with $x_{i}=-x_{j}$, and $y_{i'}$, $y_{j'}$ be the two occurrences of $x$ in $\tau$.
   Case ($a$): $x$ is neighbor-consistent, and will turn to neighbor-inconsistent by a reversal on itself. 
   Case ($b$): $x$ is neighbor-inconsistent, and will turn to neighbor-consistent by a reversal on itself. }
  \label{consistentornot}
\end{figure} 

\begin{theorem}
\label{equcnd}
Given two simple related chromosomes $\pi^*$ and $\tau$ with $dp[\pi^*]=dp[\tau]=2$,
$\pi^*=\tau$ if and only if $\mathcal{A}[\pi^*]=\mathcal{A}[\tau]$ and every repeat is neighbor-consistent. 
\end{theorem}

\begin{proof}
Assume that $\pi^*=[x_{0}, x_{1}, \dots, x_{n}, x_{n+1}]$ and $\tau=[y_{0}, y_{1}, \dots, y_{n}, y_{n+1}]$, where $x_{0}=y_{0}=+r_0$ and $x_{n+1}=y_{n+1}=-r_0$.

The sufficiency part is surely true, since we have $x_{i}=y_{i}$ for $0\leq i\leq n+1$.

Now we prove the necessity inductively.
Our inductive hypothesis is that $x_{i}=y_{i}$ for $0\leq i\leq n+1$. Initially, we have $x_{0}=y_{0}=+r_0$ and $x_1=y_1$ since $\langle r_0^t,l(x_1)\rangle = \langle r_0^t,l(y_1)\rangle$.
For the inductive step, consider $x_{i+1}$. Because $\mathcal{A}[\pi^*]=\mathcal{A}[\tau]$, the adjacency $\langle r(x_{i}), l(x_{i+1})\rangle$ appears in both $\mathcal{A}[\pi^*]$,
and $\mathcal{A}[\tau]$.
Since $|x_{i}|$ is even, the two adjacencies $\langle r(x_{i-1}), l(x_{i})\rangle$ and $\langle r(x_{i}), l(x_{i+1})\rangle$ must be matched to two adjacencies
which are associated with a single occurrence of $|x_{i}|$ in $\tau$.
From the inductive hypothesis, $x_{i-1}=y_{i-1}$ and $x_{i}=y_{i}$, $\langle r(x_{i-1}), l(x_{i})\rangle$ has been matched to $\langle r(y_{i-1}), l(y_{i})\rangle$,
Thus, $\langle r(x_{i}), l(x_{i+1})\rangle$ must be matched to $\langle r(y_{i}), l(y_{i+1})\rangle$,
together with $x_{i}=y_{i}$, we have, $l(x_{i+1})=l(y_{i+1})$ and $x_{i+1}=y_{i+1}$.
\end{proof}

Based on proposition \ref{w1} and  Theorem~\ref{equcnd}, to transform $\pi$ into $\tau$, it is sufficient to perform an odd number (at least 1) of symmetric reversals on each neighbor-inconsistent repeat,
and an even number (might be 0) of symmetric reversals on each neighbor-consistent repeat.
Hereafter, we also refer an neighbor-consistent (resp. neighbor-inconsistent) repeat as an {\it even} (resp. {\it odd}) repeat.


The main difficulty to find a sequence of symmetric reversals between
$\pi$ and $\tau$ is to choose a "correct" symmetric reversal at a time.
Note that, for a pair of occurrences $(x_i, ,x_j)$ of a repeat $x$, the orientations may be the same at present and after some reversals, the orientations of $x_i$ and $x_j$ may differ. We can only perform a reversal on a pair of occurrences of a repeat with different orientations. Thus, it is crucial to choose a "correct" symmetric reversal at the right time. In the following, We will use "intersection" graph to handle this. 

Suppose that we are given two simple related chromosomes $\pi$ and $\tau$ with $dp[\pi]=dp[\tau]=2$ and $\mathcal{A}[\pi]=\mathcal{A}[\tau]$. In this case, each repeat in the chromosomes represent an interval indicated by the two occurrences of the repeat. 
Thus, we can construct an  \emph{intersection graph} $IG(\pi, \tau)=(V[\pi], E[\pi])$.
For each repeat $x$ with $dp[\pi,x]=2$, construct a vertex $x\in V_{\pi}$, and set its weight, $\omega(x)=2$ if $x$ is even,
and $\omega(x)=1$ if $x$ is odd; set the color of $x$ black if the signs of the two occurrences of $x$ in $\pi$ are different, and white otherwise.
Construct an edge between two vertices $x$ and $y$ if and only if the occurrences of $x$ and $y$ appear alternatively in $\pi$,
i.e., let $x_{i}$ and $x_{j}$ ($i<j$) be the two occurrences of $x$, and $x_{k}$ and $x_{l}$ ($k<l$) be the two occurrences of $y$ in $\pi$, there will be an edge between the vertices $x$ and $y$ if and only if $i<k<j<l$ or $k<i<l<j$.
There are three  types of vertices in $V[{\pi}]$: black vertices (denoted as $V_{b}[\pi]$), white vertices of weight 1 (denoted as $V_{w}^{1}[\pi]$) and white vertices of weight 2 (denoted as $V_{w}^{2}[\pi]$). Thus, $V[{\pi}]=V_{b}[\pi]\cup V_{w}^{1}[\pi]\cup V_{w}^{2}[\pi]$.
In fact, the intersection graph is a circle graph while ignoring the weight and color of all the vertices. 
\begin{lemma}
\label{cannot1white}
A single white vertex of weight 1 cannot be a connected component in $IG(\pi, \tau)$ .
\end{lemma}

\begin{proof}
We prove it by contradiction. Assume that $x$ is a white vertex of weight 1, which forms a connected component of $IG(\pi, \tau)$.
Let the two occurrences of $x$ be $x_{i}$ and $x_{j}$ ($i<j$) in $\pi$, and $y_{k}$ and $y_{l}$ in $\tau$.
Since $x$ is odd, w.l.o.g, assume that $\langle r(x_{i}), l(x_{i+1})\rangle$ and $\langle r(x_{j-1}), l(x_{j})\rangle$ are matched to two adjacencies both associated with $y_{k}$.
Because $x$ is an isolated vertex in $IG(\pi, \tau)$, all the other occurrences of $|x_{i+1}|, \dots, |x_{j-1}|$ must also locate in between $x_{i}$ and $x_{j}$ in $\pi$.

Note that each adjacency is unique in $\mathcal{A}[\pi]$, as well as in $\mathcal{A}[\tau]$.
In case that $y_{k+1}$ is an occurrence of $|x_{i+1}|$, the adjacency $\langle r(y_{k+1}), l(y_{k+2})\rangle$ must be matched to an adjacency
located in between $x_{i}$ and $x_{j}$ in $\pi$,
thus, an occurrence of $|y_{k+2}|$ also locates in between $x_{i}$ and $x_{j}$ in $\pi$, so does the other occurrence of $|y_{k+2}|$ (if exist),
hence, all the adjacencies associate with the occurrences of $|y_{k+2}|$ locate in between $x_{i}$ and $x_{j}$ in $\pi$, which implies that $y_{k+3}$ exists.
The recursion can not terminate until there is some $y_{k+t}$ which is an occurrence of $|x_{j}|$, also the adjacency $\langle r(y_{k+t}), l(y_{k+t+1})\rangle=\langle r(x_{j-1}), l(x_{j})\rangle$. It is a contradiction since $\langle r(y_{k-1}), l(y_{k})\rangle=\langle r(x_{i-1}), l(x_{i})\rangle$.

The argument when $y_{k+1}$ is an occurrence of $|x_{j-1}|$ is similar.
\end{proof}


For each vertex $x$ in $IG(\pi, \tau)$, let $N(x)$ denote the set of  vertices incident to  $x$.
For a black vertex, say $x$, in $IG(\pi, \tau)$, performing a symmetric reversal of $x$ in $\pi$, yields $\pi'$, where the intersection graph $IG(\pi')=(V[\pi'], E[\pi'])$ can be derived from $IG(\pi, \tau)$ following the three rules:
\begin{itemize}
\item rule-I: for each vertex $v\in N(x)$ in $IG(\pi, \tau)$, change its color from black to white, and vice versa.
\item rule-II: for each pair of vertices $u, v\in N(x)$ of $IG(\pi, \tau)$, if $(u, v)\in E[\pi]$, then $E[\pi']=E[\pi]-\{(u, v)\}$; and if $(u, v)\notin E[\pi]$, then
$E[\pi']=E[\pi]\cup\{(u, v)\}$.
\item rule-III: subtract the weight of $x$ by one, if $\omega(x)>0$, then $V[\pi']=V[\pi]$; and if $\omega(x)=0$, then $V[\pi']=V[\pi]-\{x\}$.
\end{itemize}

\begin{figure}[htbp]
  \centering
  \includegraphics[width=0.8\textwidth]{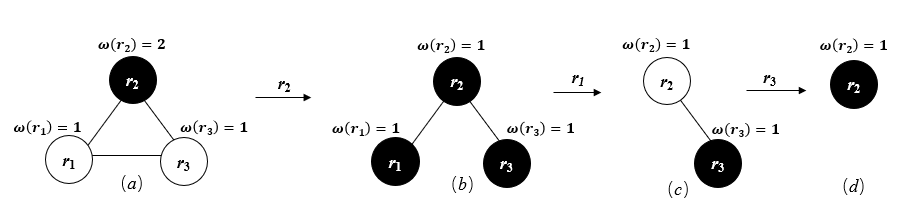}
  \caption{$\pi =[+r_0, +r_1, +1, +r_2, +r_3, +r_1, -r_2, -2, +r_3, -r_0]$, and $\tau=[+r_0, +r_1, -r_2, -2, +r_3, +r_1, +1, +r_2, +r_3, -r_0]$.
  $\mathcal{A}[\pi]=\mathcal{A}[\tau]=$$\{\langle r_0^{t}, r_1^{h}\rangle$, $\langle r_1^{t}, 1^h\rangle$, $\langle 1^t, r_2^{h}\rangle$, $\langle r_2^{t}, r_3^{h}\rangle$, $\langle r_3^{t}, r_1^{h}\rangle$, $\langle r_1^{t}, r_2^{t}\rangle$, $\langle r_2^{h}, 2^t\rangle$,  $\langle 2^{h}, r_3^{h}\rangle$, $\langle r_3^{t}, r_0^{t}\rangle\}$. The repeats $r_1$ and $r_3$ are odd, while the repeat $r_2$ is even.  ($a$) The intersection graph $IG(\pi, \tau)$, performing the symmetric reversal on repeat $r_2$ will transform $\pi$ into $\pi'=$$[+r_0$, $+r_1$, $+1$, $+r_2$, $-r_1$, $-r_3$, $-r_2$, $-2$, $+r_3$, $-r_0]$, ($b$) The intersection graph $IG(\pi',\tau)$. ($c$) The intersection graph after performing the symmetric reversal on the repeat $r_1$. ($d$) The intersection graph after performing the symmetric reversal on the repeat $r_3$.}
  \label{fig:1}
\end{figure}

If $x$ is a black vertex in $IG(\pi, \tau)$ and $\omega(x)=1$, then performing the symmetric reversal of $x$ in $\pi$ yields $\pi'$. Let $C_{1}, C_{2}, \dots, C_{m}$ be the connected components introduced by the deletion of $x$ in $IG(\pi',\tau)$, we go through some properties of performing this symmetric reversal. 
\begin{lemma}
\label{hasneighbor}
In each $C_{i}$ ($1\leq i\leq m$), there is at least one vertex $z_{i}$ such that $z_{i}\in N(x)$ in $IG(\pi, \tau)$.
\end{lemma}

\begin{proof}
Consider the scenario of $IG(\pi',\tau)$ just prior to deleting $x$, all the connected components $C_{1}, C_{2}, \dots, C_{m}$ are in a single connected component, which also contains $x$. But immediately after deleting $x$, they become $m$ separate connected components.
\end{proof}

\begin{lemma}
\label{keepcolor}
Let $x'$ be a black vertex, $\omega(x)=\omega(x')=1$, and $x'\in N(x)$ in $IG(\pi, \tau)$. After performing the symmetric reversal of $x$ in $\pi$, let $y$ be a vertex in the connected component $C_{i}$, and $x'$ is in the connected component $C_{j}$, $i\neq j$. Let $\pi''$ be the resulting chromosome after performing the symmetric reversal of $x'$ in $\pi$, then the color of $y$ is the same in $IG(\pi', \tau)$ and $IG(\pi'', \tau)$.
\end{lemma}

\begin{proof}
We conduct the proof by considering the following two cases: (1) $y\in N(x)$, (2) $y\notin N(x)$ in $IG(\pi, \tau)$.
In case (1), $(y, x')\in E[\pi]$, then the color of $y$ would be changed by performing the symmetric reversal on either $x$ or $x'$ in $\pi$.
In case (2), $(y, x')\notin E[\pi]$, then the color of $y$ would not be changed by performing the symmetric reversal on either $x$ or $x'$ in $\pi$.
\end{proof}

\begin{lemma}
\label{keepedge}
Let $x'$ be a black vertex, $\omega(x)=\omega(x')=1$, and $x'\in N(x)$ in $IG(\pi, \tau)$. After performing the symmetric reversal of $x$ in $\pi$, let $y, z$ be two vertices in the connected component $C_{i}$, and $x'$ is in the connected component $C_{j}$, $i\neq j$. Let $\pi''$ be the resulting chromosome after performing the symmetric reversal of $x'$ in $\pi$. If $(y, z)\in E[\pi']$, then $(y, z)\in E[\pi'']$.
\end{lemma}

\begin{proof}
We conduct the proof by considering the following three cases: (1) both $y, z\in N(x)$, (2) only one, say $y$, is in $N(x)$, and (3) neither of them belongs to $N(x)$ in $IG(\pi, \tau)$.
We illustrate the three cases in Figure~\ref{fig:2}.

(1) both $y, z\in N(x)$. Then $(y, z)\notin E[\pi]$, since $y$ and $x'$ are in distinct connected component,
$(y, x')\in E[\pi]$ and $(z, x')\in E[\pi]$, thus, $(y, z)\in E[\pi'']$.

(2) $y\in N(x)$, $z\notin N(x)$. Then $(y, z)\in E[\pi]$, since $y$ and $x'$ are in distinct connected component,
$(y, x')\in E[\pi]$ and  $(z, x')\notin E[\pi]$, thus, $(y, z)\in E[\pi'']$.

(3) $y\notin N(x)$, $z\notin N(x)$. Then $(y, z)\in E_{\pi}$, since $y$ and $x'$ are in distinct connected component,
$(y, x')\notin E[\pi]$ and  $(z, x')\notin E[\pi]$, thus, $(y, z)\in E[\pi'']$.
\end{proof}

\begin{figure}[htbp]
  \centering
  \includegraphics[width=0.8\textwidth]{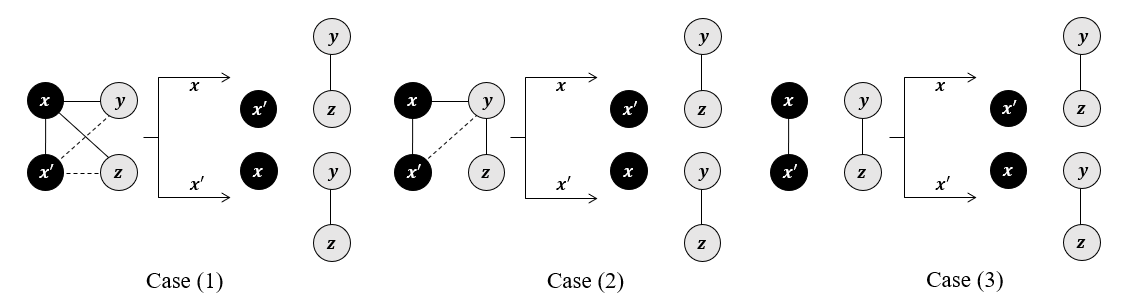}
  \caption{Three cases in the proof of Lemma~\ref{keepedge}.}
  \label{fig:2}
\end{figure}

\begin{theorem}
\label{existareversal}
If  a connected component of $IG(\pi, \tau)$ contains at least one black vertex, then there exists a symmetric reversal,
 after performing it,
 any  newly created connected component   containing a white vertex of weight 1  also  contains a black vertex.
\end{theorem}

\begin{proof}
It is apparent that performing a symmetric reversal of a black vertex $x$ of weight 2 will not introduce any new connected component, and the vertex $x$ is still black.
Assume that there is no black vertex of weight 2 in the connected component.
For each black vertex $x$, let $\Delta(x)$ be the number of white vertices in the connected components, which contains white vertices of weight 1 but not any black vertex,
introduced by performing the symmetric reversal of $x$.
Next, we show that, there must be some vertex $x$, such that $\Delta(x)=0$.

On the contrary, let $x$ be the black vertex with minimum $\Delta(x)>0$. Let $C_{1}, C_{2}, \dots, C_{m}$ be the introduced connected components after performing the symmetric reversal of $x$. There could not be a connected component which is composed of white vertices of weight 2, since otherwise, following Lemma~\ref{hasneighbor}, the neighbor of $x$ in this connected component is black prior to performing the symmetric reversal of $x$, which is in contradiction with our assumption.
W.l.o.g, assume that each of $C_{1}, C_{2}, \dots, C_{i}$ ($1\leq i\leq m$) contains a white vertex of weight 1 but no black vertex,
and each of $C_{i+1}, C_{i+2}, \dots, C_{m}$ contains at least one black vertex.

Let $x'$ be the neighbor of $x$ in $C_{1}$. Then, $x'$ is a black vertex of weight 1 prior to performing the symmetric reversal of $x$.
It is sufficient to show that $\Delta(x')<\Delta(x)$, and hence contradicting with that $\Delta(x)$ is minimum.
$\Delta(x)$ only counts the vertices in $C_{1}, C_{2}, \dots, C_{i}$.
From Lemma~\ref{keepcolor} and Lemma~\ref{keepedge}, the color of all the vertices in $C_{i+1}, C_{i+2}, \dots, C_{m}$ are preserved, and
all the edges in $C_{i+1}, C_{i+2}, \dots, C_{m}$ are preserved after performing the symmetric reversal of $x'$.
Therefore, $\Delta(x')$ will also not count any vertex in $C_{i+1}, C_{i+2}, \dots, C_{m}$,
which implies, $\Delta(x')\leq\Delta(x)$.

From Lemma~\ref{cannot1white}, $x'$ must have a neighbor in $C_{1}$. Let $x''$ be a neighbor of $x'$ in $C_{1}$.
In $IG(\pi, \tau)$, if $(x, x'')\in E_{\pi}$, then $(x', x'')\notin E_{\pi}$, and $x''$ is a black vertex prior to and after performing the symmetric reversal of $x'$.
If $(x, x'')\notin E_{\pi}$, then $(x', x'')\in E_{\pi}$, and $x''$ is a white vertex prior to and after performing the symmetric reversal of $x$,
but become a black vertex after performing the symmetric reversal of $x'$.
In either case, $\Delta(x')$ will not count $x''$, thus, $\Delta(x')<\Delta(x)$. We illustrate the proof in Figure~\ref{fig:3}.
\end{proof}

\begin{figure}[htbp]
  \centering
  \includegraphics[width=0.4\textwidth]{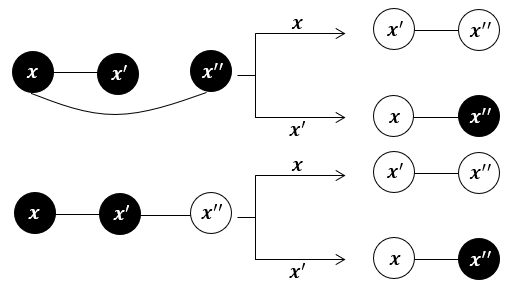}
  \caption{Whether $(x, x'')\in E_{\pi}$ or not, $x''$ will always be black after performing the symmetric reversal of $x'$, thus $\Delta(x')<\Delta(x)$.}
  \label{fig:3}
\end{figure}

The main contribution of this section is the following theorem.
\begin{theorem}
\label{determine-2}
A chromosome $\pi$ can be transformed into the other chromosome $\tau$ if and only if (I) $\mathcal{A}[\pi]=\mathcal{A}[\tau]$,
and (II) each white vertex of weight 1 belongs to a connected component of $IG(\pi, \tau)$ containing a black vertex.
\end{theorem}

\begin{proof}
$(\Rightarrow)$ If there exists an connected component of $IG(\pi, \tau)$, say $C$, which is composed of white vertices including a white vertex $x$ of wight 1,
and all the vertex in $C$ do not admit any symmetric reversal. Moreover, $x$ is odd in $\pi$.
Thus, it is impossible to find a series of symmetric reversals to make $x$ black, and then make $x$ even.
According to Theorem~\ref{equcnd}, $\pi$ can not be transformed into $\tau$.

$(\Leftarrow)$ As each odd repeat in $\pi$ corresponds to a vertex of weight 1, and each even repeat in $\pi$ corresponds to a vertex of weight 2.
Theorem~\ref{existareversal} guarantees that, each vertex of weight 1 will be reversed once, and each vertex of weight 2 will either not be reversed or be reversed twice.
Finally, all the repeats will become even, from Theorem~\ref{equcnd}, $\pi$ has been transformed into $\tau$.
\end{proof}

The above theorem implies that a breadth-first search of $IG[\pi,\tau]$ will determine whether $\pi$ can be transformed into $\tau$,
which takes $O(n^2)$ time, because $IG(\pi, \tau)$ contains at most $n$ vertices and $n^2$ edges.
We will show the details of the algorithm in Section 5, since it also serves as a decision algorithm for the general case.

In the next section, we show that the optimization version is also surprisingly polynomially solvable when the input genomes have some constraints.

\section{An Algorithm for the 2-Balanced Case of SMSR}

In this section, we consider a special case of SMSR, which we call {\em 2-Balanced}, where the duplication numbers of the two simple related chromosomes $\pi$ and $\tau$ are both 2, and $\tau$ contains both $+r$ and $-r$ for each repeat $r\in \Sigma_2$. The algorithm presented here will give a shortest
 sequence of symmetric reversals  that transform $\pi$ into $\tau$.

As $\mathcal{A}[\pi]=\mathcal{A}[\tau]$ for two related chromosomes $\pi$ and $\tau$, there is a bijection between identical adjacencies of $\mathcal{A}[\pi]$ and $\mathcal{}{A}[\tau]$.
For each pair of identical adjacencies $\langle r(x_{i}), l(x_{i+1}) \rangle\in \textit{A}[\pi]$ and $\langle r(y_{j}), l(y_{j+1}) \rangle \in \textit{A}[\tau]$,
define the sign of $\langle r(x_{i}), l(x_{i+1}) \rangle$ to be \emph{positive} if $r(x_{i})=r(y_{j})\neq l(x_{i+1})=l(y_{j+1})$,
and \emph{negative} if $r(x_{i})=l(y_{j+1})\neq l(x_{i+1})=r(y_{j})$.
We can always assume that $\langle r(x_{0}), l(x_{1}) \rangle$ and $\langle r(x_{n}), l(x_{n+1}) \rangle$ are positive, since otherwise we can reverse $\pi$ in a whole.
Note that if $r(x_{i})=l(x_{i+1})=r(y_{j})=l(y_{j+1})$, $\langle r(x_{i}), l(x_{i+1}) \rangle$ can be either positive or negative, we call it an {\em entangled} adjacency.

A segment $I$ of $\pi$ is referred to as \emph{positive} (resp. \emph{negative}) if all the adjacencies in it have positive (resp. \emph{negative}) directions,
and $I$ is \emph{maximal}, if it is not a subsegment of any other positive (resp. negative) segment in $\pi$ than $I$.
A maximal positive (resp. negative) segment is abbreviated as an \emph{MPS} (resp. \emph{MNS}). (See Figure~\ref{mps-mns} for an example.)

\begin{figure}[htbp]
  \centering
  \includegraphics[width=0.7\textwidth]{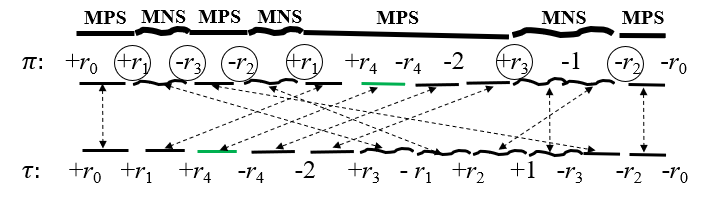}
  \caption{An example of \emph{MPS} and \emph{MNS}. Positive adjacencies and \emph{MPS} are marked with straight lines, while negative adjacencies and \emph{MNS} are marked with curved lines, the boundaries are circled, and
  the entangled adjacency $\langle r_4,-r_4\rangle$ is colored green.}
  \label{mps-mns}
\end{figure}

Now we assign proper directions to the entangled adjacencies such that the number of \emph{MNS} of $\pi$ is minimized.
Let $\langle r(x_{i}), l(x_{i+1}) \rangle$ be an entangled adjacency, then $r(x_{i})=l(x_{i+1})=r(y_{j})=l(y_{j+1})$, accordingly $l(x_{i})=r(x_{i+1})=l(y_{j})=r(y_{j+1})$.
Since $\pi$ is simple, $r(x_{i-1})\neq l(x_{i+2})$; hence, if $r(x_{i-1})=r(y_{j-1})$ and $l(x_{i+2})=l(y_{j+2})$, then both $\langle r(x_{i-1}), l(x_{i}) \rangle$ and $\langle r(x_{i+1}), l(x_{i+2}) \rangle$ are positive; if $r(x_{i-1})=l(y_{j+2})$ and $l(x_{i+2})=r(y_{j-1})$, then both $\langle r(x_{i-1}), l(x_{i}) \rangle$ and $\langle r(x_{i+1}), l(x_{i+2}) \rangle$ are negative. Thus, we have:

\begin{proposition}\label{alg111}
The two neighbors of any entangled adjacency always have the same direction.
The number of \emph{MNS} is minimized provided that the direction of each entangled adjacency is the same as its neighbors.
\end{proposition}

The two neighbors of any entangled adjacency always have the same direction.
Therefore, the number of \emph{MNS} is minimized provided that the direction of each entangled adjacency is the same as its neighbors.

Once the directions of all the adjacencies of $\pi$ are fixed, the positive segments and negative segments appear alternatively on $\pi$,
in particular, both the leftmost and rightmost segments are positive.
Let $N_{\emph{MPS}}[\pi]$ and $N_{\emph{MNS}}[\pi]$ be the number of \emph{MPS} and \emph{MNS} respectively on $\pi$, then $N_{\emph{MPS}}[\pi]-N_{\emph{MNS}}[\pi]=1$.
Accordingly,
\begin{theorem}
$\pi=\tau$ if and only if $N_{\emph{MNS}}[\pi]=0$.
\end{theorem}

\begin{proof}
The sufficient part is surely true.

We prove the efficient part inductively.
Our inductive hypothesis is that $x_{i}=y_{i}$ for $0\leq i\leq n+1$. Initially, we have $x_{0}=y_{0}=+0$.
For the inductive step,
let $y_{j}$ be the other occurrence of $|y_{i}|$, since $y_{i}$ and $y_{j}$ have distinct signs in $\tau$, $r(y_{i})=l(y_{j})$,
together with the inductive hypothesis that $x_{i}=y_{i}$, we have, $r(x_{i})=r(y_{i})=l(y_{j})$.
As the adjacency $\langle r(x_{i}), l(x_{i+1})\rangle$ is positive,
there must be some $y_{k}$ of $\tau$, such that $r(y_{k})=r(x_{i})$ and $l(y_{k+1})=l(x_{i+1})$.
Since $y_{i}$ is the unique occurrence satisfying $r(y_{i})=r(x_{i})$, then $k=i$, and consequently, $l(y_{i+1})=l(x_{i+1})$ and $x_{i+1}=y_{i+1}$.
This completes the induction.
\end{proof}

An occurrence of $x_{i}$ ($1\leq i\leq n$), is called a {\em boundary} if the directions of its left adjacency $\langle r(x_{i-1}), l(x_{i}) \rangle$
and right adjacency $\langle r(x_{i}), l(x_{i+1}) \rangle$ are distinct. Thus, any boundary is shared by an \emph{MPS} and an \emph{MNS}. (See Figure~\ref{mps-mns} for an example.)
Let $x_{i}$ and $x_{j}$ ($i<j$) be a pair of occurrences of $|x_{i}|$ in $\pi$ with distinct signs, performing a symmetric reversal $\rho(i, j)$ on $|x_{i}|$ yields $\pi'$.
We have,
\begin{lemma}
\label{atmost1}
$N_{\emph{MNS}}[\pi]-N_{\emph{MNS}}[\pi']\leq1$.
\end{lemma}

\begin{proof}
The proof is conducted by enumerating all possible cases that $x_{i}$ and $x_{j}$ may be boundaries, in the interior of a \emph{MPS} or a \emph{MNS}.
For a revenue of 1, the unique case is that both $x_{i}$ and $x_{j}$ are boundaries and the substring $[x_{i},\dots, x_{j}]$ starts and ends with both \emph{MNS} or both \emph{MPS}.
\end{proof}

Lemma~\ref{atmost1} implies that a scenario will be optimal, provided that each symmetric reversal subtracts the number of \emph{MNS} by one.
Luckily, we can always find such symmetric reversals till $\pi$ has been transformed into $\tau$.

\begin{lemma}
\label{3-properties}
Let $x_{i}$ and $x_{j}$ ($i<j$) be the two occurrences of the repeat $x$ in $\pi$, (I) either they are both boundaries or none of them is a boundary.
(II) If $x_{i}$ and $x_{j}$ are both boundaries, the adjacencies $\langle r(x_{i}), l(x_{i+1}) \rangle$ and $\langle r(x_{j-1}), l(x_{j}) \rangle$ have the same direction,
and the adjacencies $\langle r(x_{i-1}), l(x_{i}) \rangle$ and $\langle r(x_{j}), l(x_{j+1}) \rangle$ have the same direction.
(III) Moreover, if $x_{i}=x_{j}$, then $\langle r(x_{i}), l(x_{i+1}) \rangle$ and $\langle r(x_{j-1}), l(x_{j}) \rangle$ are matched to two consecutive adjacencies in $\tau$, and $\langle r(x_{i-1}), l(x_{i}) \rangle$ and $\langle r(x_{j}), l(x_{j+1}) \rangle$ are matched to two consecutive adjacencies in $\tau$.
\end{lemma}

\begin{proof}
(I) Without loses of generality, assume that $x_{i}$ is a boundary.
Let $y_{k}$ and $y_{l}$ be the occurrences of $|x_{i}|$ in $\tau$ with distinct signs, thus $l(y_{k})=r(y_{l})\neq r(y_{k})=l(y_{l})$.
We partition the adjacencies associated with $y_{k}$ and $y_{l}$ into two groups: $\langle r(y_{k-1}), l(y_{k}) \rangle$ and $\langle r(y_{l-1}), l(y_{l}) \rangle$ form Group-(I)
and $\langle r(y_{k}), l(y_{k+1}) \rangle$ and $\langle r(y_{l}), l(y_{l+1}) \rangle$ form Group-(II). The group partition guarantees that $l(y_{k})$ and $r(y_{l})$ are in different groups, and $r(y_{k})$ and $l(y_{l})$ are also in different groups.
Since the two adjacencies $\langle r(x_{i-1}), l(x_{i}) \rangle$ and $\langle r(x_{i}), l(x_{i+1}) \rangle$ have distinct directions, so they are either matched to Group-(I) or Group-(II), accordingly, the two adjacencies $\langle r(x_{j-1}), l(x_{j}) \rangle$ and $\langle r(x_{j}), l(x_{j+1}) \rangle$ have to be matched to the other group.
In either case, $\langle r(x_{j-1}), l(x_{j}) \rangle$ and $\langle r(x_{j}), l(x_{j+1}) \rangle$ have distinct directions. Thus, $x_{j}$ is a boundary.

(II) If $\langle r(x_{i-1}), l(x_{i}) \rangle$ and $\langle r(x_{i}), l(x_{i+1}) \rangle$ are matched to group-(I), then $\langle r(x_{i}), l(x_{i+1}) \rangle$ is negative,
$l(x_{j})$ must be matched to either $r(y_{k})$ or $r(y_{l})$, which implies that $\langle r(x_{j-1}), l(x_{j}) \rangle$ is also negative.
If $\langle r(x_{i-1}), l(x_{i}) \rangle$ and $\langle r(x_{i}), l(x_{i+1}) \rangle$ are matched to group-(II), then $\langle r(x_{i}), l(x_{i+1}) \rangle$ is positive,
$l(x_{j})$ must be matched to either $l(y_{k})$ or $l(y_{l})$, which implies that $\langle r(x_{j-1}), l(x_{j}) \rangle$ is also positive.
Since the direction of $\langle r(x_{i-1}), l(x_{i}) \rangle$ is different from $\langle r(x_{i}), l(x_{i+1}) \rangle$, and the direction of $\langle r(x_{j}), l(x_{j+1}) \rangle$ is different from $\langle r(x_{j-1}), l(x_{j}) \rangle$, $\langle r(x_{i-1}), l(x_{i}) \rangle$ and $\langle r(x_{j}), l(x_{j+1}) \rangle$ have the same direction.

(III) When $x_{i}=x_{j}$, the positive pair of adjacencies must be matched to the two adjacencies associated with the occurrence with sign ``+'', and the negative pair of adjacencies must be matched to the two adjacencies associated the occurrence with sign ``-'', so they are consecutive in $\tau$.
\end{proof}

\begin{figure}[htbp]
  \centering
  \includegraphics[width=0.7\textwidth]{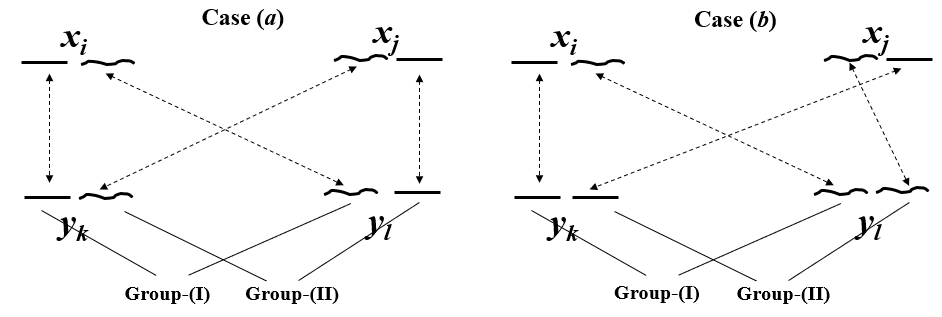}
  \caption{An illustration of Lemma~\ref{3-properties}. Assume that $x_{i}=y_{k}=-y_{l}$, $\langle r(x_{j-1}), l(x_{j}) \rangle$ and $\langle r(x_{j}), l(x_{j+1}) \rangle$
  are matched to Group-(I). In case ($a$), the signs of $x_{i}$ and $x_{j}$ are distinct; and in case ($b$), the signs of $x_{i}$ and $x_{j}$ are identical.}
  \label{bdrmatching}
\end{figure}

\begin{lemma}
\label{same-segment}
Every \emph{MPS} in $\pi$ has an identical segment in $\tau$, and every \emph{MNS} in $\pi$ has a reversed and negated segment in $\tau$. 	
\end{lemma}

\begin{proof}
Let $I=[x_{i},x_{i+1},\dots, x_{i+k}]$ be a \emph{MPS} in $\pi$.
There must be an adjacency $\langle r(y_{j}), l(y_{j+1}) \rangle$ such that $r(x_{i})=r(y_{j}
)$ and $l(x_{i+1})=l(y_{j+1})$,
We conduct an inductive proof. Our inductive hypothesis is that $x_{i+t}=y_{j+t}$, for all $0\leq t\leq k$.
Initially, we have $x_{i}=y_{j}$. For the inductive step, let $y_{j'}$ be the other occurrence of $|y_{j+t}|$ in $\tau$.
Since $y_{j'}$ and $y_{j+t}$ have distinct signs, $r(y_{j+t})=l(y_{j'})$, together with the inductive hypothesis, we have $r(x_{i+t})=r(y_{j+t})=l(y_{j'})$,
which implies that $\langle r(x_{i+t}), l(x_{i+t+1}) \rangle$ must be matched to $\langle r(x_{j+t}), l(y_{i+t+1}) \rangle$.
Consequently $l(x_{i+t+1})=l(y_{i+t+1})$ and $x_{i+t+1}=y_{i+t+1}$.
This completes the induction.

The proof for an \emph{MNS} in $\pi$ is similar hence omitted.
\end{proof}

\begin{theorem}
\label{properboundaries}
If there exists an MNS in $\pi$, then there exists a pair of boundaries, which are occurrences of the same repeat  with different orientations.
\end{theorem}

\begin{proof}
Assume to the contrary that each pair of boundaries, which are duplications of the same repeat, have the same sign.
From Lemma~\ref{3-properties} and~\ref{same-segment}, each pair of boundaries that have equal absolute values locate on two distinct \emph{MPS}s,
whose corresponding segments are adjacent in $\tau$.
As each \emph{MPS} is associated with two boundaries, there is a path along all the \emph{MPS}s and boundaries in $\pi$, the corresponding segments of all the \emph{MPS}s forms a substring of $\tau$.
Note that $\langle r(x_{0}), l(x_{1}) \rangle$ and $\langle r(x_{n}), l(x_{n+1}) \rangle$ are both positive,
so the substring contains both $y_{0}$ and $y_{n+1}$, and it becomes the whole of $\tau$.
That is a contradiction, since $\tau$ also contains the corresponding segments of the \emph{MNS}s.
\end{proof}

Now, we formally present the algorithm as Algorithm 1 for computing the least number of symmetric reversals to transform $\pi$ into $\tau$.

\begin{algorithm}[tb]
\caption{An algorithm for the 2-Balanced Case of SMSR }
\hspace*{0.02in}{\bf Input:}: Two chromosomes $\pi$ and $\tau$, with $\mathcal{A}[\pi]=\mathcal{A}[\tau]$ and $dp[\pi]=dp[\tau]=2$, and the occurrences of each repeat has different orientation  on $\tau$.\\
\hspace*{0.02in}{\bf Output:}: a sequence of minimum number of symmetric reversals that transforms $\pi$ into $\tau$.

\begin{algorithmic}[1] 
    \STATE Delete the redundant repeats from $\pi$ and $\tau$ when $dp[\pi]=2$.(Lemma~\ref{norepetitive}) 
    \STATE Build the bijection between $\mathcal{A}[\pi]$ and $\mathcal{A}[\tau]$.
    \STATE Determine the direction of each adjacency (except entangled) of $\pi$ according to the bijection.
    \STATE Assign the direction of each entangled adjacency   the same as  its two  neighbors. (See Proposition \ref{alg111}, where the two neighbors always have the same direction.)
    \STATE Identify all the \emph{MPS}, \emph{MNS} and boundaries in $\pi$.
    \WHILE{there exists an \emph{MNS}}
         \STATE Find a pair of boundaries of the same repeat with  different orientations.
         \STATE Perform a symmetric reversal  on  this pair of inverted repeats.
    \ENDWHILE\\
\RETURN{} all the symmetric reversals.
\end{algorithmic}
\end{algorithm}

\begin{theorem}
Algorithm 1  gives  a sequence of minimum number of symmetric reversals that transforms $\pi$ into $\tau$, and runs in $O(n^2)$ time.
\end{theorem}

\begin{proof}
Theorem~\ref{properboundaries} guarantees the existence of two boundaries with $x_{i}=-x_{j}$, whenever there is an \emph{MNS}.
From Lemma~\ref{3-properties}-(II), performing the symmetric reversal on $|x_{i}|$ will decrease the number of \emph{MNS} by one.
Thus, the number of symmetric reversals performed by the algorithm \emph{Balanced-2 SMSR} is equal to the number of  \emph{MNS} in $\pi$,
which is optimum according to Theorem~\ref{atmost1}.

As for the time complexity, it takes $O(n^2)$ time to build the bijection between $\mathcal{A}[\pi]$ and $\mathcal{A}[\tau]$.
Obviously, steps 2-4 take linear time. The {\bf while}-loop runs at most $n/2$ rounds, and in each round, it takes linear time to find a pair of proper boundaries.
Totally, the time complexity is $O(n^2)$.
\end{proof}

We comment that the result in this section is perhaps the best that we can hope, since   we will show in section 6 that, when the duplication number is 2 but $\tau$ is not required to contain both $+|x|$ and $-|x|$ for each repeat $|x|\in \Sigma_2$, the \emph{SMSR} problem becomes NP-hard. More generally, when the duplication number of $\pi$ and $\tau$ is unlimited, the scenario is quite different. Nevertheless, in the next section, we proceed to handle this general case.


\section{An $O(n^2)$ Decision Algorithm for the General Case}

For the general case, i.e., when the duplication number for the two related input genomes is arbitrary,
the extension of the algorithm in Section 3 is non-trivial as it is impossible to make the genomes
simple. Our overall idea is to fix any bijection $f$ between the (identical) adjacencies of the input genomes, 
and build the corresponding alternative-cycle graph. This alternative-cycle graph is changing according to the 
corresponding symmetric reversals; and we show that, when the graph contains only 1-cycles, then the target 
$\tau$ is reached. Due to the changing nature of the alternative-cycle graph, we construct a blue edge intersection graph
to capture these changes. However, this is not enough as the blue intersection graph built from the alternative-cycle 
graph could be disconnected and we need to make it connected by adding additional vertices such that the resulting 
sequence of symmetric reversals are consistent with the original input genomes, and can be found in the new intersection
graph (called IG, which is based on the input genomes $\pi$ and $\tau$ as well as $f$). We depict the details in the following.

Suppose that we are given two related chromosomes $\pi=[x_{0}, x_{1}, \dots, x_{n+1}]$ and $\tau=[y_{0}, y_{1}, \dots, y_{n+1}]$, such that $x_{0}=y_{0}=+r_0$ and $x_{n+1}=y_{n+1}=-r_0$.
Theorem~\ref{necessary} shows that $\mathcal{A}[\pi]=\mathcal{A}[\tau]$ is a necessary condition,
thus there is a bijection $f$ between identical adjacencies in $\mathcal{A}[\pi]$ and $\mathcal{A}[\tau]$, as shown in Figure~\ref{acg}.
Based on the bijection $f$ , we construct the alternative-cycle graph $ACG(\pi,\tau, f)$  as follows.
For each $x_{i}$ in $\pi$, construct an ordered pair of nodes, denoted by $l(x_{i})$ and $r(x_{i})$, which are connected by a red edge.
For each $y_{k}$ in $\tau$, assume that $\langle r(y_{k-1}), l(y_{k})\rangle$ is matched to $\langle r(x_{i-1}), l(x_{i})\rangle$, and $\langle r(y_{k}), l(y_{k+1})\rangle$ is matched to $\langle r(x_{j-1}), l(x_{j})\rangle$, in the bijection $f$. There are four cases:
\begin{enumerate}
\item  $l(y_{k})=l(x_{i})$ and $r(y_{k})=r(x_{j-1})$, then connect $l(x_{i})$ and $r(x_{j-1})$ with a blue edge,
\item  $l(y_{k})=r(x_{i-1})$ and $r(y_{k})=r(x_{j-1})$, then connect $r(x_{i-1})$ and $r(x_{j-1})$ with a blue edge,
\item  $l(y_{k})=l(x_{i})$ and $r(y_{k})=l(x_{j})$, then connect $l(x_{i})$ and $l(x_{j})$ with a blue edge,
\item  $l(y_{k})=r(x_{i-1})$ and $r(y_{k})=l(x_{j})$, then connect $r(x_{i-1})$ and $l(x_{j})$ with a blue edge.
\end{enumerate}
\begin{figure}[htbp]
  \centering
  \includegraphics[width=0.7\textwidth]{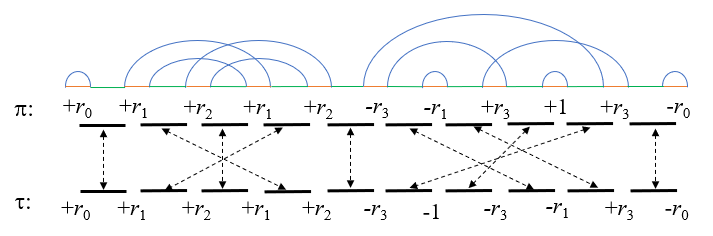}
  \caption{The bijection between identical adjacencies in $\mathcal{A}[\pi]$ and $\mathcal{A}[\tau]$, and the corresponding alternative-cycle graph.}
  \label{acg}
\end{figure}

Actually, two nodes connected by a red edge implies they are from the same occurrence of some repeat/gene in $\pi$, so each occurrence of some repeat/gene in $\pi$ corresponds to a red edge; and similarly, two nodes connected by a blue edge implies that they are from the same occurrence of some repeat/gene in $\tau$, thus each occurrence of some repeat/gene in $\tau$ corresponds to a blue edge.
Note that each node associates with one red edge and one blue edge, so $ACG(\pi,\tau, f)$ is composed of edge disjoint cycles, on which the red edges and blue edge appears alternatively.
A cycle composed of $c$ blue edges as well as $c$ red edges is called a $c$-cycle, it is called a long cycle when $c\geq2$.

\begin{theorem}
\label{all1-cycle}
Given two chromosomes $\pi^*$ and $\tau$, $\pi^*=\tau$ if and only if $\mathcal{A}[\pi^*]=\mathcal{A}[\tau]$, and there is a bijiection $f$ between the identical adjacencies in $\mathcal{A}[\pi^*]$ and $\mathcal{A}[\tau]$, such that all the cycles in the resulting alternative-cycle graph $ACG(\pi^*, \tau, f)$ are 1-cycles.
\end{theorem}

\begin{proof}
Assume that $\pi^*=[z_{0}, z_{1}, \dots, z_{n}, z_{n+1}]$ and $\tau=[y_{0}, y_{1}, \dots, y_{n}, y_{n+1}]$, where $z_{0}=y_{0}=+r_0$ and $z_{n+1}=y_{n+1}=-r_0$.

The sufficiency part is surely true, since we have $z_{i}=z_{i}$ and bijection $f$
is to match $\langle r(z_{i-1}),l(z_i)\rangle$ with $\langle r(y_{i-1}),l(y_i)\rangle$.

Now we prove the necessity part inductively.
Our inductive hypothesis is that $z_{i}=y_{i}$ for $0\leq i\leq n+1$. Initially, we have $l(z_{0})=l(y_{0})=r_0^{h}$ and
$r(z_{0})=r(y_{0})=r_0^{t}$, as $\mathcal{A}[\pi^*]=\mathcal{A}[\tau]$, we always assume that $z_{1}=y_{1}$, since otherwise, we can reverse the whole chromosome $\tau$.
For the inductive step, as $l(z_{i})$ and $r(z_{i})$ form a 1-cycle, so then adjacency $\langle r(z_{i}), l(z_{i+1})\rangle$ is matched to $\langle r(y_{i}), l(y_{i+1})\rangle$,
and from the inductive hypothesis that $r(z_{i})=r(y_{i})$, we have, $l(z_{i+1})=l(y_{i+1})$, and thus, $z_{i+1}=y_{i+1}$.
\end{proof}

The above theorem gives us a terminating condition for our algorithm: let $\pi$ and $\tau$ be the input chromosomes,
and our algorithm keeps updating the alternative-cycle graph until all cycles in it become 1-cycles.
unfortunately, in the following, we observe that some cycles can not be performed on symmetric reversals directly, then we consider these cycles intersecting with each other as a connected component. But this is still not enough, since there could also be some connected components which do not admit any symmetric reversal, we managed to handle this case by joining all the cycles of the same repeat into a whole connected component. 

\begin{lemma}
\label{sameelement}
In an  alternative-cycle graph, each cycle corresponds to a unique repeat and  every edge (both red and blue)    in the cycle  corresponds to  an  occurrence of the unique repeat.
\end{lemma}

\begin{proof}
W.l.o.g, assume that $l(x_{i})$ are $l(x_{j})$ connected with a blue edge, from the construction of the alternative-cycle graph,
there must be an occurrence in $\tau$, say $y_{k}$, such that $\{l(x_{i}), l(x_{j}\}=\{l(y_{k}), r(y_{k})\}$,
thus, $|x_{i}|=|x_{j}|=|y_{k}|$, and the blue edge $(l(x_{i}), l(x_{j}))$ corresponds the occurrence $y_{k}$ of the repeat $|y_{k}|$.
\end{proof}

Since each gene appears once in $\pi$, Lemma~\ref{sameelement} implies that each gene has a 1-cycle in $ACG(\pi,\tau,f)$, these 1-cycles will be untouched throughout our algorithm.
\begin{lemma}
\label{a path}
In an alternative-cycle graph, if we add a {\it green}  edge  connecting each pairs of nodes $r(x_{i})$ and $l(x_{i+1})$ (for all $0\leq i\leq n$), then all the blue edges and green edges together form a (blue and green  alternative) path. (See Figure 6.)
\end{lemma}

\begin{proof}
Actually, the green edge connecting $r(x_{i})$ and $l(x_{i+1})$ ($0\leq i\leq n$) is the adjacency $\langle r(x_{i}), l(x_{i+1})\rangle$ of $\mathcal{A}[\pi]$, which is identical to some adjacency $\langle r(y_{j}), l(y_{j+1})\rangle$ of $\mathcal{A}[\tau]$ according to the bijection between identical adjacencies of $\mathcal{A}[\pi]$ and $\mathcal{A}[\tau]$.
Therefore, $y_{j}$ and $y_{j+1}$ appears consecutively in $\tau$,
and following the construction of $ACG(\pi,\tau, f)$ and Lemma~\ref{sameelement}, they correspond to the two blue edges, one of which is associated with $r(x_{i})$ and the other is associated with $l(x_{i+1})$ in $ACG(\pi,\tau, f)$, thus, the two blue edges are connected through the green edge $(r(x_{i}), l(x_{i+1}))$.
The above argument holds for every green edge, therefore, all the blue edges and green edges constitute a path. We show an example in Figure~\ref{acg}.
\end{proof}

Let $x\in \Sigma$ be a repeat.
Let $x_i$ and $x_{j}$ be two occurrences of $x$ in $\pi$, where $i\neq j$.
A blue edge is \emph{opposite} if it connects $l(x_{i})$ and $l(x_{j})$ or $r(x_{i})$ and $r(x_{j})$.
A blue edge is \emph{non-opposite} if it connects $l(x_{i})$ and  $r(x_{j})$ or $r(x_{i})$ and $l(x_{j})$.

Specially, the blue on any 1-cycle (with a blue edge  and a red edge) is non-opposite.
A cycle is \emph{opposite} if it contains at least one opposite  blue edge.

\begin{lemma}
\label{samesign}
Let $x_i$ and $x_{j}$ be two occurrences of repeat $x$ in $\pi$. In the alternative-cycle graph $ACG(\pi,\tau,f)$,
if $l(x_{i})$ and $l(x_{j})$ ( or $r(x_{i})$ and $r(x_{j})$) are connected with an opposite edge, $x_i$ and $x_j$ has different orientations;
and if $l(x_{i})$ and $r(x_{j})$ ( or $r(x_{i})$ and $l(x_{j})$) are connected with a non-opposite edge, $x_i$ and $x_j$ has the same orientations.
\end{lemma}

\begin{proof}
We conduct the proof according to the construction of the alternative-cycle graph $ACG(\pi,\tau,f)$.

If $l(x_{i})$ and $l(x_{j})$ are connected with an opposite edge, there must be an occurrence
of $x$ in $\tau$, say $y_{k}$, such that $\{l(x_{i}), l(x_{j})\}=\{l(y_{k}), r(y_{k})\}$.
If $l(x_{i})=l(y_{k})$ and $l(x_{j})=r(y_{k})$,
then the orientations of $x_{i}$ and $y_{k}$ are the same, and the orientations of $x_{j}$ and $y_{k}$ are different, thus $x_i$ and $x_j$ have different orientations.
If $l(x_{i})=r(y_{k})$ and $l(x_{j})=l(y_{k})$, then the orientations of $x_{i}$ and $y_{k}$ are different, and the orientations of $x_{j}$ and $y_{k}$ are the same, also $x_i$ and $x_j$ have different orientations.
It will be similar when $r(x_{i})$ and $r(x_{j})$ are connected with an opposite edge.

If $l(x_{i})$ and $r(x_{j})$ are connected with a non-opposite edge, there must be an occurrence
of $x$ in $\tau$, say $y_{k}$, such that $\{l(x_{i}), r(x_{j}\}=\{l(y_{k}), r(y_{k})\}$.
If $l(x_{i})=l(y_{k})$ and $r(x_{j})=r(y_{k})$,
then both $x_{i}$ and $x_{j}$ have the same orientations as $y_{k}$, thus $x_i$ and $x_j$ have the same orientations.
If $l(x_{i})=r(y_{k})$ and $r(x_{j})=l(y_{k})$, then both $x_{i}$ and $x_{j}$ have different orientations with $y_{k}$, thus $x_i$ and $x_j$ have the same orientations.
It will be similar when $r(x_{i})$ and $l(x_{j})$ are connected with an opposite edge.
\end{proof}

\begin{proposition}
\label{joinorsolitclycle}
Given a $k$-cycle $C$ of $x$, performing a symmetric reversal on two occurrences of $x$ that are connected by an opposite blue edge, will break $C$ into a $(k-1)$-cycle as well as a 1-cycle. Given a $k_{1}$-cycle $C_{1}$ and a $k_{2}$-cycle $C_{2}$ of $x$, performing a symmetric reversal on the two occurrences of $x_{i}\in C_{1}$ and $x_{j}\in C_{2}$,
will join $C_{1}$ and $C_{2}$ into a $(k_{1}+k_{2})$-cycle.
\end{proposition}

Now, we construct the \emph{blue edge intersection graph} $BG(\pi, \tau,f)=(BV_{\pi}, BE_{\pi}, f)$ according to $ACG(\pi,\tau, f)$, viewing each blue edge as an interval of the two nodes it connects. For each interval, construct an original vertex in $BV_{\pi}$, and set its weight to be 1, set its color to be black if the blue edge is opposite, and white otherwise.
An edge in $BE_{\pi}$ connects two vertices if and only if their corresponding intervals intersect but neither overlaps the other. An example of the blue edge intersection graph is shown in Figure~\ref{final-IG}-($b$).

Note that each connected component of $BG(\pi, \tau,f)$ forms an interval on $\pi$, for each connected component $P$ in $BG(\pi, \tau,f)$, we use $\overline{P}$ to denote its corresponding interval on $\pi$.
\begin{lemma}
\label{includewholered}
Let $P$ be some connected component of $BG(\pi, \tau,f)$, the leftmost endpoint of $\overline{P}$ must be a left node of some $x_{i}$, i.e., $l(x_{i})$,
and the rightmost endpoint of $\overline{P}$ must be a right node of some $x_{j}$, i.e., $r(x_{j})$, where $i<j$.
\end{lemma}

\begin{proof}
Since each blue edge connects two nodes of the interval $\overline{P}$, the number of nodes in $\overline{P}$ must be even. Hence the boundary nodes of an interval can neither be both left nodes nor be both right nodes.

Assume to the contrary that the leftmost endpoint of $\overline{P}$ is $r(x_{i})$ and the rightmost endpoint of $\overline{P}$ is $l(x_{j})$. There are a total of $2(j-i)$ nodes appearing in $\overline{P}$. If we connect each pairs of nodes $r(x_{k})$ and $l(x_{k+1})$ (for all $i\leq k\leq j$) with a green edge, then there are $j-i$ green edges as well as $j-i$ blue edges between the $2(j-i)$ nodes of $\overline{P}$. Since each node is associated with one green edge and one blue edge, these green edges and blue edges form cycles, which is a contradiction to Lemma~\ref{a path}.
\end{proof}

\begin{lemma}
\label{samecycle}
All the vertices in $BG(\pi, \tau,f)$  corresponding to the blue edges on the same  long cycle in $ACG(\pi, \tau,f)$ are in the same connected component of $BG(\pi, \tau,f)$.
\end{lemma}

\begin{proof}
Assume to the contrary that there exist some connected components each containing a part of blue edges of some cycle.
Let $P$ be such a type of connected component that $\overline{P}$ does not overlap any other interval of such type of connected components.
Let $e=(x_{i}^{t}, x_{j}^{h})$ and $e'=(x_{i'}^{t}, x_{j'}^{h})$ be two blue edges on a cycle $C$, such that their corresponding vertices $v\in P$, but $v'\notin P$.
From the way we choose $P$, the two nodes $x_{i'}^{t}$ and $x_{j'}^{h}$ connected by $e'$ in $C$ can not be both inside the interval of $P$.
So they must be both outside the interval of $P$.
Also, the two nodes $x_{i}^{t}$ and $x_{j}^{h}$ connected by $e$ in $C$ can not be both on the boundary of the interval of $P$, since otherwise,
$e$ will not intersect with any other blue edge corresponding to a vertex of $P$.
W.l.o.g, assume that $x_{i}^{t}$ appears inside the interval of $P$.
Since $e$ and $e'$ are on the same cycle $C$, besides $e$, there must be an alternative path from $x_{i}^{t}$ to $x_{j'}^{h}$ on $C$.
From Lemma~\ref{includewholered}, it is impossible that the travel along an alternative path outside the scope of an interval via a red edge,
so the alternative path from $x_{i}^{t}$ to $x_{j'}^{h}$ must contain a blue edge which connects a node inside $\overline{P}$ with a node outside $\overline{P}$,
which contracts that $P$ is a connected component.
\end{proof}

As the two blue edges of a non-opposite 2-cycle do not intersect each other, we have,
\begin{corollary}
\label{atleast3}
A non-opposite 2-cycle can not form a connected component of $BG(\pi, \tau,f)$.
\end{corollary}

For each repeat $x$, assume that it constitutes $k$ cycles in $ACG(\pi, \tau, f)$. Let $x_{i_1}$, $x_{i_2}$, $\dots$, $x_{i_k}$ be the $k$ occurrences of $x$ that are in distinct cycles in $ACG(\pi, \tau, f)$, where $1\leq i_1<i_2<\cdots<i_k\leq n$.
We construct $k-1$ \emph {additional vertices} corresponding to the intervals $[r(x_{i_{j}})-\epsilon, l(x_{i_{j+1}})+\epsilon]$ to $BV_{\pi}$ ($1\leq j\leq k-1$), for each such vertex, set its weight to be 1, and set its color to be black if the signs of $x_{i_{j}}$ and $x_{i_{j+1}}$ are distinct, and white otherwise. See the vertex marked with 10 in Figure~\ref{final-IG}-($c$) for an example. Also, there is an edge between two vertices of $BV_{\pi}$ if and only if their corresponding intervals intersect, but none overlaps the other.
The resulting graph is called the \emph{intersection graph} of $\pi$, denoted as $IG(\pi, \tau, f)=(V[\pi], E[\pi])$. An example is shown in Figure~\ref{final-IG}-($c$).
Let $V_{\pi}^w\subseteq V[\pi]$ be the subset of vertices which corresponding to non-opposite blue edges on long cycles in $ACG(\pi,\tau,f)$.  

\begin{figure}[htbp]
  \centering
  \includegraphics[width=0.7\textwidth]{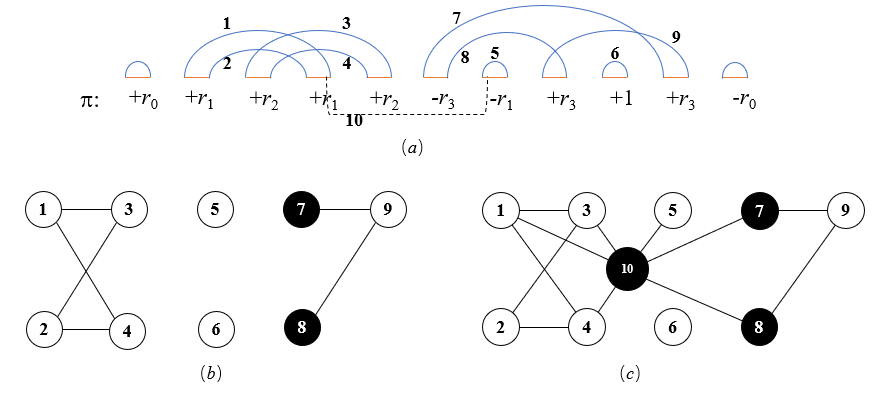}
  \caption{($a$) The alternative-cycle graph $ACG(\pi, \tau, f)$, where each blue edge is marked with a number. ($b$) The blue edge intersection graph $BG(\pi, \tau,f)$. ($c$) The intersection graph $IG(\pi, \tau, f)$ with additional vertices, where each number represents an interval.}
  \label{final-IG}
\end{figure}

From Lemma~\ref{includewholered} and the construction of the intersection graph of $\pi$, all the vertices corresponding to all the blue edges of the same repeat are in the same connected component.
Note that the intersection graph of $\pi$ may be distinct, when the bijection between identical adjacencies of $\mathcal{A}[\pi]$ and $\mathcal{A}[\tau]$ differs.
Nevertheless, we have,
\begin{lemma}
\label{sameCC}
Let $\pi$ and $\tau$ be two related chromosomes with $\mathcal{A}[\pi]=\mathcal{A}[\tau]$.
Let $x_{i}$ and $x_{j}$ ($i<j$) be two occurrences of $x$ in $\pi$, and $x_{i'}$ and $x_{j'}$ ($i'<j'$) be two occurrences of $x'$ in $\pi$, if either $i<i'<j<j'$ or $i'<i<j'<j$ is satisfied,
then, based on any bijection $f$ between $\mathcal{A}[\pi]=\mathcal{A}[\tau]$, in the intersection graph $IG(\pi, \tau, f)$, the vertices corresponding to all the intervals of $x$ and $x'$ are in the same connected component.   
\end{lemma}

\begin{proof}
Assume that $i<i'<j<j'$, since the other case is similar.
From Lemma~\ref{includewholered} and the construction of the intersection graph,
in any intersection graph of $\pi$, the vertices corresponding to the intervals associated with $x_{i}$ and $x_{j}$ are in the same connected component $P$;
also, the vertices corresponding to the intervals associated with $x_{i'}$ and $x_{j'}$ are in the same connected component $P'$.
We show that $\overline{P}$ intersects $\overline{P'}$, and thus $P=P'$.

It is impossible that $\overline{P}$ and $\overline{P'}$ are disjoint, because the interval $[l(x_{i}), \dots, r(x_{j})]$ is a part of $\overline{P}$,
and the interval $[l(x_{i'}), \dots, r(x_{j'})]$ is a part of $\overline{P'}$, but $[l(x_{i}), \dots, r(x_{j})]$ intersects with $[l(x_{i'}), \dots, r(x_{j'})]$.

If $\overline{P}$ overlaps $\overline{P'}$, in the intersection graph $IG(\pi, \tau, f)$,
there exists a path from some interval associated with $x_{j}$ to some interval associated with $x_{i}$,
among the intervals corresponding to the vertices on this path, there must be one that intersects with $\overline{P'}$.

If $\overline{P'}$ overlaps $\overline{P}$, in the intersection graph $IG(\pi, \tau, f)$, there exists a path from some interval associated with $x_{i'}$ to some interval associated with $x_{j'}$,
among the intervals corresponding to the vertices on this path, there must be one that intersects with $\overline{P}$.
\end{proof}

Actually, the connected components of the intersection graph partition the repeats on $\pi$ into groups.
From Lemma~\ref{includewholered} and Lemma~\ref{sameCC}, the group partition of the repeats is independent of the bijection between identical adjacencies of $\mathcal{A}[\pi]$ and $\mathcal{A}[\tau]$. In other words, the group partition will be fixed once $\pi$ and $\tau$ are given.



Similar to the intersection graph of chromosomes with a duplication number of 2,
the intersection graph of chromosomes with unrestricted duplication number also admit the rule-I, rule-II, and rule-III, as in Section 3, while performing a symmetric reversal on $\pi$.

\begin{theorem}
\label{existgoodreversal}
If a connected component of $IG(\pi, \tau, f)$ contains a black vertex, then there exists a symmetric reversal, after performing it, we obtain $\pi'$, 
any newly created connected component containing a white vertex, which corresponds to a blue edge on a non-opposite long cycle in $ACG(\pi', \tau, f)$,  also contains a black vertex.
\end{theorem}

\begin{proof}
Each vertex in $IG(\pi, \tau, f)=(V[\pi], E[\pi])$ corresponds to an interval that are flanked by two occurrences of the same repeat.
A vertex is black when the two occurrences have different signs,
thus admits a symmetric reversal.
For each black vertex $x$, let $C_{1}, C_{2}, \dots, C_{m}$ be the newly created connected components after performing the symmetric reversal on $x$,
where each of $C_{1}, C_{2}, \dots, C_{i}$ contains a white vertex corresponding to a blue edge on a non-opposite long cycle in $ACG(\pi', \tau, f)$, but does not contain a black vertex, 
each of $C_{i+1}, C_{i+2}, \dots, C_{j}$ contains additional vertices and vertices corresponding to blue edges on 1-cycles in $ACG(\pi', \tau, f)$, but does not black vertex,
and each of $C_{j+1}, C_{j+2}, \dots, C_{m}$ contains a black vertex, $0\leq i\leq j\leq m$.
Let $\Delta(x)$ be the number of white vertices in $C_{1}, C_{2}, \dots, C_{j}$.

We show next that if there exists a connected component, which contains a white vertex corresponding to a blue edge on a non-opposite long cycle in $ACG(\pi', \tau, f)$ , but does not contain a black vertex, then there would be a black vertex $x'$ such that $\Delta(x')<\Delta(x)$.

Assume to the contrary that $x$ is the black vertex with the minimum $\Delta(x)>0$. 
Let $x'$ be the neighbor of $x$ in $C_{1}$. Then, $x'$ is a black vertex prior to performing the symmetric reversal of $x$.
It is sufficient to show that $\Delta(x')<\Delta(x)$, and hence contradicting with that $\Delta(x)$ is minimum.
From Lemma~\ref{keepcolor} and Lemma~\ref{keepedge}, the color of all the vertices in $C_{j+1}, C_{j+2}, \dots, C_{m}$ are preserved, and
all the edges in $C_{j+1}, C_{j+2}, \dots, C_{m}$ are preserved after performing the symmetric reversal of $x'$.
Therefore, $\Delta(x')$ will also not count any vertex in $C_{j+1}, C_{j+2}, \dots, C_{m}$,
consequently, $\Delta(x')\leq\Delta(x)$.

Since $C_{1}$ contains a white vertex corresponding to a blue edge on a non-opposite long cycle, following Lemma~\ref{samecycle} and Corollary~\ref{atleast3},
$C_{1}$ contains at least 3 vertices. Let $N_{1}(x)$ be the neighbor set of $x$ in $C_{1}$, and $x'\in N_{1}(x)$.
There must exist a neighbor of $x'$ in $C_{1}$, say $x''$. 
If $x''\in N_{1}(x)$, then $(x',x'')\notin E[\pi]$, and $x''$ is black prior to performing the symmetric reversal on $x$,
thus, performing a symmetric reversal on $x'$ instead of $x$ will keep $x''$ be black, which implies $\Delta(x')<\Delta(x)$.
On the other side, if $x''\notin N_{1}(x)$, then $(x',x'')\in E[\pi]$, and $x''$ is white prior to performing the symmetric reversal on $x$,
Hence performing a symmetric reversal on $x'$ instead of $x$ will make $x''$ be black, which also implies $\Delta(x')<\Delta(x)$.
Both contract the assumption that $\Delta(x)$ is the minimum.
\end{proof}

\begin{theorem}
\label{finalcon}
A chromosome $\pi$ can be transformed into the other  chromosome $\tau$ by symmetric reversals if and only if (I) $\mathcal{A}[\pi]=\mathcal{A}[\tau]$, and (II) each white  vertex in $V_{\pi} ^w$ belongs to a connected component of $IG(\pi, \tau, f)$ containing  a black vertex.
\end{theorem}

\begin{proof}
$(\Rightarrow)$ If there exists a connected component of $IG(\pi, \tau, f)$, say $C$, which contains a  white  vertex in $V_{\pi} ^w$ but does not contain a black vertex, then all the vertices of $C$ are white,
and they do not admit any symmetric reversal. Moreover, there must be a white vertex, say $x$, which corresponds to a blue edge of a non-opposite long cycle in $ACG(\pi,\tau,f)$.
Thus, it is impossible to find a series of symmetric reversals to transform this long cycle into 1-cycles.
According to Theorem~\ref{all1-cycle}, $\pi$ can not be transformed into  $\tau$.

$(\Leftarrow)$ Theorem~\ref{existgoodreversal} guarantees that, each original vertex either can be performed with a symmetric reversal or its corresponding blue edge has been in a 1-cycle.
Finally, all the blue edges which correspond to the original vertices are in 1-cycles; following Theorem~\ref{all1-cycle}, $\pi$ has been transformed into $\tau$.
\end{proof}

Now, we are ready to formally present the decision algorithm based on Theorem \ref{finalcon} for both the general case and the case, where  the  duplication number 2 in Algorithm 2. We just directly test conditions (I) and (II) in Theorem \ref{finalcon}. Note that each connected component in $IG(\pi, \tau,f)$
may contain more than one black vertex.
By setting  $Q = V_{\pi}^b$ in line 11, we can guarantee that each connected component in $IG(\pi, \tau, f)$ is explored once during the breadth first search so that $O(n^2)$ running time can be kept.
\begin{algorithm}[tb]
\caption{The decision algorithm for SSR}
\textbf{Input}: Two related chromosomes $\pi$ and $\tau$.\\
\textbf{Output}: $YES\backslash NO$
\begin{algorithmic}[1] 
\IF {$\mathcal{A}[\pi]\neq \mathcal{A}[\tau]$}
\STATE return NO.
\ENDIF
\STATE Delete the redundant repeats from $\pi$ and $\tau$ when $dp[\pi]=2$.(Lemma~\ref{norepetitive}) 
\STATE Build a bijection $f$ between $\mathcal{A}[\pi]$ and $\mathcal{A}[\tau]$ when $dp[\pi]>2$ by mapping identical adjacencics  
together and adopting any bijection among multiple occurences of the same adjacency.
\STATE Construct the alternative-cycle graph $ACG(\pi, \tau, f)$ based on $f$
when $dp[\pi]>2$.\STATE Construct the corresponding intersection graph $IG(\pi, \tau, f)=(V_{\pi}, E_{\pi})$ (based on $ACG(\pi, \tau, f)$ when $dp[\pi]>2$ and directly from $\pi$ and $\tau$ if $dp[\pi]=2$).
 \IF { $dp[\pi]=2$ }
 \STATE set $V_{\pi}^w$ to be the set of weight 1 white vertices
\ENDIF
 \IF { $dp[\pi]>2$}
 \STATE set $V_{\pi}^w$ to be the set of white vertices corresponding to non-opposite blue edge on long cycle.
\ENDIF
\STATE Let $V_{\pi}^b$ be set of black vertices in $IG(\pi, \tau,f)$.
\STATE Set queue $Q = V_{\pi}^b$.
\STATE  Do a breadth first search using $Q$ as the initial value to mark all vertices in the same component of every $x\in V_{\pi}^b$.
\IF {there exists an $u\in V_{\pi}^w$, which is not marked}
   \RETURN{} NO.
\ENDIF\\
\RETURN{} YES.
\end{algorithmic}
\end{algorithm}

\noindent{\bf Running time of Algorithm 2:}
Let us analyze the time complexity of Algorithm 2.
Verifying whether $\mathcal{A}[\pi]=\mathcal{A}[\tau]$ can be done in $O(n^2)$ time.
It takes $O(n^2)$ time to build an bijection between $\mathcal{A}[\pi]$ and $\mathcal{A}[\tau]$, and construct the cycle graph $ACG(\pi, \tau, f)$, as well as the corresponding intersection graph $IG(\pi, \tau,f)$.
It remains to analyze the size of $IG(\pi, \tau,f)$. For each repeat, say $x$, there are $dp[x,\pi]$ original vertices and $c[x]-1$ additional vertices in $IG[\pi]$, where $c[x]$ is the number of cycles of $x$ in $ACG(\pi,\tau,f]$. Note that $c[x]\leq dp[x,\pi]$ and $\sum_{x\in\sum}dp[x,\pi]=n$. Thus, the total number of vertices in $IG(\pi,\tau,f)$ is bounded by $\sum_{x\in\sum}(dp[x,\pi]+c[x]-1)\leq 2\sum_{x\in\sum}dp[x,\pi]-1=2n-1$, then the number of edges in $IG[\pi]$ is at most $4n^2$.
The whole breadth-search process takes $O(n^2)$ time, since there are at most $2n-1$ vertices and at most $4n^2$ edges in $IG(\pi, \tau,f)$. Therefore, Algorithm 2 runs in $O(n^2)$ time.

\section{Hardness Result for SMSR}
In this section, we show that the optimization problem \emph{SMSR} is NP-hard.
When we initially investigated SMSR on the special case when the input genomes have duplication number 2,
we found that it is somehow related to searching a Steiner set on the intersection graph, which is a type of circle graph with its vertices colored and weighted.
We review the definitions on circle graphs as follows.

\begin{definition}
A graph is a \emph{circle graph} if it is the intersection graph of chords in a circle. 
\end{definition}
\begin{definition}
A graph is an \emph{overlap graph}, if each vertex corresponds to an interval on the real line, and there is an edge between two vertices if and only if their corresponding intervals intersect but one does not contain the other.
\end{definition}
The graph class of overlap graphs is equivalent to circle graphs if we join the two endpoints of the real line into a circle.
Circle graph can be recognized in polynomial time \cite{GHS}.

\begin{definition}{Minimum Steiner Tree Problem: }

\textbf{Instance:}~~A connected graph $G$, a non-empty subset $X\subseteq V(G)$, called terminal set.

\textbf{Task:}~~Find a minimum size subset $S\subseteq V(G)\backslash X$ such that the induced subgraph $G[S\cup X]$ is connected.
\end{definition}
 More than thirty five years ago, Johnson stated that the \emph{Minimum Steiner Tree} problem on \emph{Circle Graphs} was in $P$ \cite{Johnson} and referred to a personal communication as a reference. However, to date, there is no published   polynomial algorithm to solve the problem. Figueiredo et al. revisited Johnson's table recently, and although they still marked the problem as in  $P$, the reference remains ``ongoing'' \cite{FMSS}.
Here as a by-product,  we  clarify that  the \emph{Minimum Steiner Tree} problem on \emph{Circle Graphs} is NP-hard.
The reduction is from \emph{MAX-(3,B2)-SAT}, in which each clause is of size exactly 3 and each variable occurs exactly twice in its positive and twice in the negative form.
\begin{theorem}
\emph{MAX-(3,B2)-SAT} is NP-hard \emph{\cite{Berman}}.
\end{theorem}
Next, we conduct a reduction from \emph{MAX-(3,B2)-SAT} to the \emph{Minimum Steiner Tree} problem on \emph{Circle Graphs}.
\begin{theorem}
\label{inappr3sat}
The \emph{Minimum Steiner Tree} problem on \emph{Circle Graphs} is NP-hard.
\end{theorem}

\begin{proof}
Given an instance of \emph{MAX-(3,B2)-SAT} with $n$ variables $\{x_{1}, x_{2}, ..., x_{n}\}$ and $m$ clauses $\{c_{1}, c_{2}, ..., c_{m}\}$,
we construct a \emph{circle graph} denoted by $OLG(3,B2)$ as follows.

\begin{itemize}
\item \emph{Variable Intervals}. For each variable $x_{i}$, let $q_{i}=300(i-1)$. Construct four group of ladder intervals, where $B[i]_{a}^b=[q_{i}+50(a-1)+4(b-1)+1, q_{i}+50(a-1)+4(b-1)+7]$, for $a=1,2,3,4$ and $b=1,2,3,4,5,6$. Let $b[i]=\{B[i]_{a}^{b}|a=1,2,3,4$, and $b=2,5\}$, and $B[i]=\{B[i]_{a}^{b}|a=1,2,3,4$, and $b=1,3,4,6\}$.
    Then construct four intervals: $P[i]_{1}=[q_{i}+12, q_{i}+125]$, which intersects with both $B[i]_{1}^3$ and  $B[i]_{3}^6$; $P[i]_{2}=[q_{i}+62, q_{i}+175]$, which intersects with both $B[i]_{2}^3$ and $B[i]_{4}^6$; $N[i]_{1}=[q_{i}+3, q_{i}+53]$, which intersects with both $B[i]_{1}^1$ and $B[i]_{2}^1$; and $N[i]_{2}=[q_{i}+116, q_{i}+166]$, which intersects with both $B[i]_{3}^4$ and $B[i]_{4}^4$. Let $P[i]=\{P[i]_{1}, P[i]_{2}\}$, and $N[i]=\{N[i]_{1}, N[i]_{2}\}$.
\item \emph{Clause Intervals}. For each clause $c_{a}$ ($1\leq a\leq m$), construct an interval $c_{a}=[300n+3an, 300n+3an+2n+1]$. We still use $\mathbb{C}=\{c_{1}, c_{2}, \dots, c_{m}\}$ to denote the set of intervals constructed by the clauses.
\item \emph{Positive Literal Intervals}.
If $x_{i}$ appears in $c_{j}$ and $c_{k}$($j<k$) as positive literals, construct two intervals $f[i]^{j}=[q_{i}+200, q_{i}+210]$ and $f[i]^{k}=[q_{i}+220, q_{i}+230]$, which are independent with all the previous intervals; and construct $G[i]^j=[q_{i}+209, 300n+3jn+i]$, which intersects with $f[i]^{j}$ and $c_{j}$;
$G[i]^k=[q_{i}+229, 300n+3kn+i]$, which intersects with $f[i]^{k}$ and $c_{k}$.
Construct $D[i]^{j}=[q_{i}+25, q_{i}+201]$, which intersects with $f[i]^{j}$ and $B[i]_{1}^6$, and $D[i]^{k}=[q_{i}+75, q_{i}+221]$, which intersects with $f[i]^{k}$ and $B[i]_{2}^6$.
\item \emph{Negative Literal Intervals}. If $x_{i}$ appears in $c_{j'}$ and $c_{k'}$($j'<k'$) as negative literals, construct two intervals $f[i]^{j'}=[q_{i}+240, q_{i}+250]$ and $f[i]^{k'}=[q_{i}+260, q_{i}+270]$, which are independent with all the previous intervals; and construct $G[i]^{j'}=[q_{i}+249, 300n+3j'n+i]$, which intersects with $f[i]^{j'}$ and $c_{j'}$, and $G[i]^{k'}=[q_{i}+269, 300n+3k'n+i]$, which intersects with $f[i]^{k'}$ and $c_{k'}$.
Construct $D[i]^{j'}=[q_{i}+103, q_{i}+241]$, which intersects with $f[i]^{j'}$ and $B[i]_{1}^6$, and $D[i]^{k'}=[q_{i}+153, q_{i}+261]$, which intersects with $f[i]^{k'}$ and $B[i]_{2}^6$.
Let $f[i]=\{f[i]^{j}, f[i]^{k}, f[i]^{j'}, f[i]^{k'}\}$ and let $D[i]=\{D[i]^{j}, D[i]^{k}, D[i]^{j'}, D[i]^{k'}\}$.
\item \emph{Subtree Intervals}.
For each variable $x_{i}$, construct four intervals $t[i]=\{t[i]_{1}, t[i]_{2}, t[i]_{3}, t[i]_{4} \}$, where $t[i]_{a}=[-(4i+a)+4, q_{i}+50a-1]$ for $a=1,2,3,4$.
Finally, construct two intervals $r=\{r_{1}=[-4n-1, 301n], r_{2}=[-4n-2,0]\}$.
\end{itemize}

Let $b=\cup_{1}^{n}b[i]$, $f=\cup_{1}^{n}f[i]$, $t=\cup_{1}^{n}t[i]$, $\mathbb{B}=\cup_{1}^{n}B[i]$, $\mathbb{P}=\cup_{1}^{n}P[i]$, $\mathbb{N}=\cup_{1}^{n}N[i]$,  $\mathbb{D}=\cup_{1}^{n}D[i]$, $\mathbb{G}=\cup_{1}^{n}G[i]$, $t\cup_{1}^{n}t[i]$.
Let $OLG(3,B2)$ be the corresponding circle graph of all the constructed intervals, define the input terminal set of the \emph{Steiner Tree Problem} as the vertices corresponding to intervals in $b\cup f\cup c\cup t\cup r$, and the vertices corresponding to the intervals in $\mathbb{B}\cup\mathbb{P}\cup \mathbb{N}\cup \mathbb{D}\cup \mathbb{G}$ are candidate Steiner vertices.

It is not hard to observe that the vertices corresponding to the intervals in $t\cup r$ have already formed a subtree,
and all the candidate Steiner vertices are connect to this subtree, thus it does not matter whether they are connected themselves or not.
The terminal vertices corresponding to the intervals in $b\cup f\cup c$ are mutually independent, then they must connect to the subtree through Steiner vertices.

We show that the \emph{MAX-(3,B2)-SAT} instance is satisfiable if and only if the \emph{minimum Steiner tree} problem on $OLG(3,B2)$ has an optimum solution of $14n$ vertices.

For each variable $x_{i}$, assume that $c_{j}$ and $c_{k}$ are the two clauses where $x_{i}$ appears as positive literals,
and $c_{j'}$ and $c_{k'}$ are the two clauses where $x_{i}$ appears as negative literals.

($\Rightarrow$) Assume that there is a truth assignment to the \emph{MAX-(3,B2)-SAT} instance,
we can obtain a solution of $14n$ intervals for the \emph{Minimum Steiner Tree} problem as follows.

If $x_{i}$ is assigned true, by selecting the 6 vertices corresponding to intervals $\{P[i]_{1}$, $P[i]_{2}$, $D[i]_{j'}$, $D[i]_{k'}$, $G[i]_{j}$, $G[i]_{k}\}$,
and 8 vertices corresponding to the ladder intervals $\{B[i]_{1}^3$, $B[i]_{2}^3$, $B[i]_{3}^6$, $B[i]_{4}^6$, $B[i]_{1}^4$, $B[i]_{2}^4$, $B[i]_{3}^1$, $B[i]_{4}^1\}$,
all the 12 terminals vertices corresponding interval in $b\cup f$, as well as the terminals corresponding to $c_{j}$ and $c_{k}$ will connect to the subtree.

If $x_{i}$ is assigned false, by selecting the 6 vertices corresponding to intervals $\{N[i]_{1}$, $N[i]_{2}$, $D[i]_{j}$, $D[i]_{k}$, $G[i]_{j'}$, $G[i]_{k'}\}$,
and 8 vertices corresponding to the ladder intervals $\{B[i]_{1}^1$, $B[i]_{2}^1$, $B[i]_{3}^4$, $B[i]_{4}^4$, $B[i]_{1}^6$, $B[i]_{2}^6$, $B[i]_{3}^3$, $B[i]_{4}^3\}$,
all the 12 terminals vertices corresponding interval in $b\cup f$, as well as the terminals corresponding to $c_{j'}$ and $c_{k'}$ will connect to the subtree.
Therefore, we obtain a Steiner set of size $14n$.

($\Leftarrow$) Assume that the \emph{Minimum Steiner Tree} problem on $OLG(3,B2)$ has a Steiner set of $14n$ vertices, we show that there is a truth assignment to the \emph{MAX-(3,B2)-SAT} instance.

Firstly, for each variable $x_{i}$,
we have, constraint (I): in order to connect the 8 vertices corresponding to intervals in $b[i]$ to the subtree, the Steiner set includes at least 8 vertices corresponding to intervals in $B[i]$; and constraint (II): in order to connect the 4 vertices corresponding to intervals in $f[i]$ to the subtree,
the Steiner set must include at least 4 vertices corresponding to intervals in $D[i]\cup G[i]$.
But any 12 vertices satisfying the above constraints are not enough,
because the four terminals corresponding to the intervals in $\{B[i]_{1}^{2}, B[i]_{2}^{2}, B[i]_{3}^{5}, B[i]_{4}^{5}\}$ are still not connected to the subtree.
For each of them, say $B[i]_{1}^{2}$ for an example, if it must be connected to the subtree through $B[i]_{1}^{3}$ and $P[i]_{1}$, $B[i]_{1}^{1}$ and $N[i]_{1}$, or $B[i]_{1}^{3}$ and $B[i]_{1}^{4}$.
If it chooses $B[i]_{1}^{3}$ and $B[i]_{1}^{4}$, then the terminal corresponding to $B[i]_{1}^{5}$ has to connect to the subtree through $B[i]_{1}^{6}$ and $D[i]^j$,
this implies both neighbors of the vertex corresponding to $B[i]_{1}^{5}$ are selected in the Steiner set.
Similar argument holds for each of $\{B[i]_{1}^{2}, B[i]_{2}^{2}, B[i]_{3}^{5}, B[i]_{4}^{5}\}$.
Thus, to connect the four terminals corresponding to the intervals in $\{B[i]_{1}^{2}, B[i]_{2}^{2}, B[i]_{3}^{5}, B[i]_{4}^{5}\}$ to the subtree,
besides a neighbor of each vertex, each of them also needs an additional vertex, though two of them may share an additional vertex.
Hence the Steiner set must include at least another two vertices, which should either be $P[i]_{1}$ and $P[i]_{2}$, or $N[i]_{1}$ and $N[i]_{2}$.
Therefore, since the Steiner set is of size $14n$, for each variable $x_{i}$($1\leq i\leq n$), it includes exactly 14 vertices corresponding to intervals
in $B[i]\cup P[i]\cup N[i]\cup D[i]\cup G[i]$; moreover, either the vertices corresponds to intervals of $P[i]$ or the vertices corresponds to intervals of $N[i]$ are selected.

If the Steiner set includes $P[i]_{1}$ and $P[i]_{2}$, to connect the two vertices corresponding to $B[i]_{3}^{2}$ and $B[i]_{4}^{2}$ to the subtree,
it also includes $D[i]^{j'}$ and $D[i]^{k'}$, then to connect the two vertices corresponding to $f[i]^{j}$ and $f[i]^{k}$,
the Steiner set could contains one of $D[i]^{j}$ and $G[i]^{j}$, and one of $D[i]^{k}$ and $G[i]^{k}$,
we can revise the Steiner set such that it includes $G[i]^{j}$ and $G[i]^{k}$; subsequently, the vertices corresponding to the intervals $c_{j}$ and $c_{k}$ are connected to the subtree.
In this case, we assign $x_{i}$ true to satisfy $c_{j}$ and $c_{k}$.

If the Steiner set contains $N[i]_{1}$ and $N[i]_{2}$, to connect the two vertices corresponding to $B[i]_{1}^{5}$ and $B[i]_{2}^{5}$ to the subtree,
it also contains $D[i]^{j}$ and $D[i]^{k}$, then to connect the two vertices corresponding to $f[i]^{j'}$ and $f[i]^{k'}$,
the Steiner set could contain one of $D[i]^{j'}$ and $G[i]^{j'}$, and one of $D[i]^{k'}$ and $G[i]^{k'}$.
Hence we can revise the Steiner set such that it contains $G[i]^{j'}$ and $G[i]^{k'}$, subsequently, the vertices corresponding to the intervals $c_{j'}$ and $c_{k'}$ are connected to the subtree.
In this case, we assign $x_{i}$ false to satisfy $c_{j'}$ and $c_{k'}$.
\end{proof}

We take as instance of \emph{MAX-(3,B2)-SAT} as an example, where $\{x_{1}, x_{2}, x_{3}\}$ is the set of the variables and $\{c_{1}=\{x_{1}, x_{2}, x_{3}\}, c_{2}=\{x_{1}, \bar{x}_{2}, \bar{x}_{3}\}, c_{3}=\{\bar{x}_{1}, \bar{x}_{2}, x_{3}\}$, $c_{4}=\{\bar{x}_{1}, x_{2}, \bar{x}_{3}\}\}$ is the set of clauses. The constructed instance and the corresponding circle graph are in Fig.~\ref{intervals} and Fig.~\ref{olg}
respectively.
\begin{figure}[H]
\centering
\includegraphics[width=1.0\textwidth]{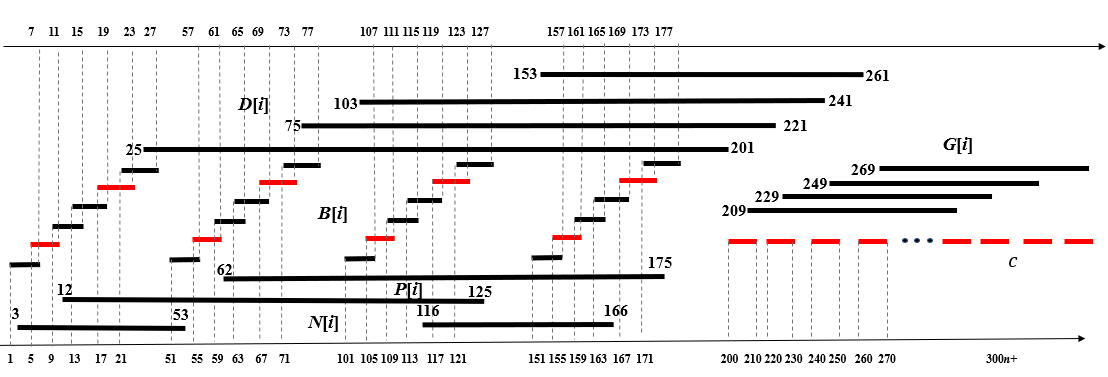}
\caption{The intervals base on $x_{1}$. Intervals corresponding to terminal vertices are red, and intervals corresponding to Steiner vertices are colored black.}
\label{intervals}
\end{figure}
\begin{figure}[H]
\centering
\includegraphics[width=0.9\textwidth]{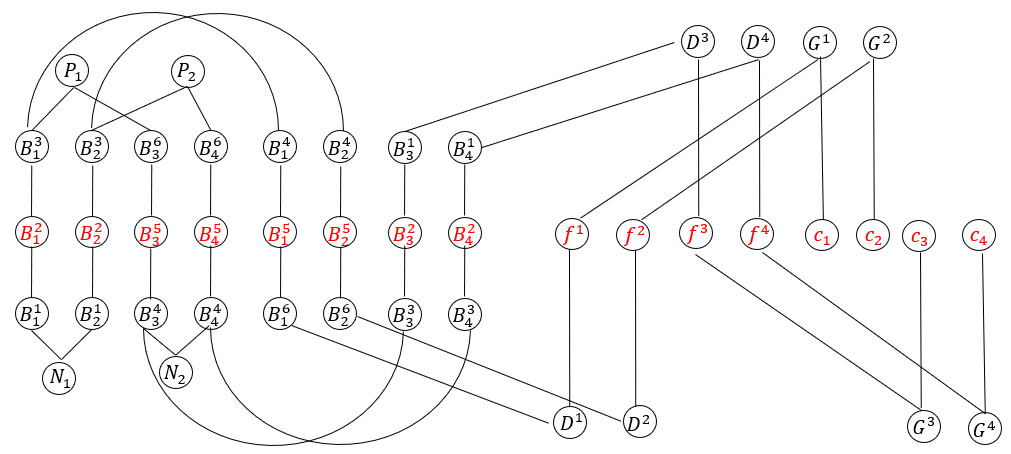}
\caption{The corresponding circle graph according to the intervals in Figure~\ref{intervals}.}
\label{olg}
\end{figure}

Finally, we conduct a reduction from the Minimum Steiner Tree problem on circle graphs to the SMSR problem even if the duplication number of the converted chromosomes is 2.

\begin{lemma}
  \label{twocasecycle}
Given two  related simple chromosomes $\pi$ and $\tau$, both have a duplication number 2, let $f$ be the unique bijection between $\mathcal{A}[\pi]$ and $\mathcal{A}[\tau]$. 
$ACG(\pi,\tau,f)$ is composed of 1-cycles and 2-cycles, and each appearance of an even repeat corresponds to a 1-cycle,
and the two appearances of every odd repeat corresponds to a 2-cycle.
\end{lemma}

\begin{proof}
Since the duplication number of $\pi$ and $\tau$ is 2, from Lemma~\ref{sameelement}, each cycle of $ACG(\pi,\tau,f)$ contains at most 2 red edges.

Let $x$ be an even repeat and $x_{i}$, $x_{j}$ be its two occurrences in $\pi$. Since $x$ is neighbor-consistent, there must be two occurrences of $x$, say $y_{k}$ and $y_{l}$ in $\tau$ such that $\langle r(x_{i-1}), l(x_{i}) \rangle$ and $\langle r(x_{i}), l(x_{i+1}) \rangle$ are matched to $\langle r(x_{k-1}), l(x_{k}) \rangle$ and $\langle r(x_{k}), l(x_{k+1}) \rangle$, and $\langle r(x_{j-1}), l(x_{j}) \rangle$ and $\langle r(x_{j}), l(x_{j+1}) \rangle$ are matched to $\langle r(x_{l-1}), l(x_{l}) \rangle$ and $\langle r(x_{l}), l(x_{l+1}) \rangle$ in the bijection $f$. From the construction of $ACG(\pi,\tau,f)$, $l(x_{i})$ and $r(x_{i})$ are connected by a blue edge, and $l(x_{j})$ and $r(x_{j})$ are connected by a blue edge, which implies that 
$x_{i}$ and $x_{j}$ correspond to 1-cycles respectively.

Let $x$ be an odd repeat and $x_{i}$, $x_{j}$ be its two appearances in $\pi$. Since $x$ is neighbor-inconsistent, there must be two occurrences of $x$, say $y_{k}$ and $y_{l}$ in $\tau$ such that $\langle r(x_{i-1}), l(x_{i}) \rangle$
is matched to an adjacency involving $y_{k}$ while $\langle r(x_{i}), l(x_{i+1}) \rangle$ is matched to
an adjacency involving $y_{l}$, thus $l(x_{i})$ and $r(x_{i})$ are not connected by a blue edge, which implies that $x_{i}$ and $x_{j}$ are in a 2-cycle.
\end{proof}

\begin{theorem}
The SMSR problem is NP-hard even if the input genomes have duplication number 2.
\end{theorem}

\begin{proof}
Given a circle graph $G=(V,E)$, where $X\subseteq V$ is the set of terminal vertices,
for each vertex $v\in V$, let $[l(v), r(v)]$ be its corresponding interval on the real line. 
Assume that $l(v_{min})$ is the minimum and $r(v_{max})$ is the maximum.   

We first construct two 1-cycles: $(l(v_{min})-1-\delta, l(v_{min})-1+\delta)$, which corresponds to ``+0''; and
$(r(v_{max})+1-\delta, r(v_{max})+1+\delta)$, which corresponds to ``-0''.

For each terminal vertex $t\in X$, we construct two intersecting non-opposite 2-cycles: one is $(l(t)-\delta, l(t)+\delta, r(t)-\delta, r(t)+\delta, l(t)-\delta)$, where $(l(t)-\delta$, $l(t)+\delta)$ and $(r(t)-\delta$, $r(t)+\delta)$ are red edges and $(l(t)-\delta$, $r(t)+\delta)$ and $(l(t)+\delta$, $r(t)-\delta)$ are blue edges. Let $t_{1}$ be the repeat corresponding to this cycle, the two occurrences both have a ``+''sign in $\pi$; the other is $(r(t)-3\delta, r(t)-2\delta, r(t)+2\delta, r(t)+3\delta, r(t)-3\delta)$, where $(r(t)-3\delta$, $r(t)-2\delta)$ and $(r(t)+2\delta$, $r(t)+3\delta)$ are red edges and $(r(t)-3\delta$, $r(t)+3\delta)$ and $(r(t)-2\delta$, $r(t)+2\delta)$ are blue edges. Let $t_{2}$ be the repeat corresponding to this cycle, the two occurrences both have a ``+''sign in $\pi$. 
Then, we denote $l(t)-\delta$ and $r(t)-\delta$ by $t_{1}^h$,
$l(t)+\delta$ and $r(t)+\delta$ by $t_{1}^t$, $r(t)-3\delta$ and $r(t)+2\delta$ by $t_{2}^h$, $r(t)-2\delta$ and $r(t)+3\delta$ by $t_{2}^t$. 

Choose an arbitrary terminal vertex $x\in X$, we construct an opposite 2-cycle: $(r(x)-5\delta, r(x)-4\delta, r(x)+5\delta, r(x)+4\delta, r(x)-5\delta)$, where $(r(x)-5\delta$, $r(x)-4\delta)$ and $(r(x)+4\delta$, $r(x)+5\delta)$ are red edges and $(r(x)-5\delta$, $r(x)+4\delta)$ and $(r(x)-4\delta$, $r(x)+5\delta)$ are blue edges. Let $x_{3}$ be the repeat corresponding to this cycle, the occurrence corresponding to the red edge $(r(x)-5\delta, r(x)-4\delta)$ has a ``+'' sign, and the occurrence corresponding to the red edge $(r(x)+4\delta, r(x)+5\delta)$ has a ``-'' sign in $\pi$.
Then, we denote $r(x)-5\delta$ and $r(x)+\delta$ by $x_{3}^h$,
$r(x)-4\delta$ and $r(x)+4\delta$ by $x_{3}^t$.

For each candidate Steiner vertex $s$, we construct two 1-cycles $(l(s)-\delta, l(s)+\delta)$, $r(s)-\delta, r(s)+\delta)$. Let $s_{1}$ be the repeat corresponding to this cycle, the two occurrences both have a ``+'' sign in $\pi$.
Then, we denote $l(s)-\delta$ and $r(s)-\delta$ by $s_{1}^h$,
$l(s)+\delta$ and $r(s)+\delta$ by $s_{1}^t$.

Let the graph constructed above be denoted by $ACG(G,X)$,
Next, we show that $ACG(G,X)$ is a well-defined  alternative-cycle graph, i.e., there exist two related simple chromosomes $\pi$ and $\tau$ with $\mathcal{A}[\pi]=\mathcal{A}[\tau]$ (then there is a bijection $f$), such that $ACG(\pi,\tau,f)= ACG(G,X)$.

In $ACG(G,X)$, each red edge corresponds to an occurrence of a repeat, thus clearly the chromosome $\pi$ is a sequence of the occurrences. Then we can view the blanks between red edges as adjacencies.
From the above construction, each blue edge connects in the cycles corresponding to $v_{i}$, then connects node $v_{i}^h$ with node $v_{i}^t$, thus an blue edge will correspond to an occurrence of $v_{i}$ in $\tau$. 
Therefore, from Lemma~\ref{a path}, it is sufficient to show that all the blanks and blue edges form a path. 

Sorting the vertices of $V$ in an increasing order of $l(v)$,
then $ACG(G,X)$ is obtained by adding a 1-cycle, two intersecting non-opposite 2-cycles, and an opposite 2-cycle iteratively. We prove that the blanks and blue edges from a path inductively. Initially, the two blue edges of the cycles $(0^h, 0^t)$ and $(0^t, 0^h)$ surely form a path. Assume that we have a path $\mathbb{P}$ of blanks
and blue edges, $\mathbb{P}=(0^h, n_{1}, \dots, n_{k},0^h)$. 

(1) in case that a 1-cycle $(s_1^h, s_1^t)$ is added in between the blank of $n_{i}$ and $n_{j}$, if $j=i+1$, then $\mathbb{P}'=(0^h, \dots, n_{i}, s_1^h, s_1^t, n_{j}, \dots, n_{k})$ is a path;  if $j=i-1$, then $\mathbb{P}'=(n_{0}, \dots, n_{j}, s_1^t, s_1^h, n_{i}, \dots, n_{k}, 0^h)$ is also a path. 

(2) in case that two intersecting non-opposite 2-cycles $(t_{1}^h, t_{1}^t, t_{1}^h, t_{1}^t, t_{1}^h)$ and $(t_{2}^h, t_{2}^t, t_{2}^h, t_{2}^t, t_{2}^h)$ are added in between the two blanks of $n_{i}$, $n_{j}$ and $n_{i'}$, $n_{j'}$. We just show the case that $i+1=j<i'=j'-1$, since all other cases are similar. 
$\mathbb{P}'=(0^h, \dots$, $n_{i}$, $t_{1}^h$, $t_{1}^t$, $t_{2}^h$, $t_{2}^t$, $t_{1}^h$, $t_{1}^t$, $n_{j}$, $\dots$, $n_{i'}$, $t_{2}^h$, $t_{2}^t$, $n_{j'}$, $\dots$, $n_{k},0^h)$ is also a path. Specially, if two intersecting non-opposite 2-cycles are added in between one blank, the new path can be obtained by deleting the segment $(n_{j}$, $\dots$, $n_{i'})$ from $\mathbb{P}'$.

(3) in case that an opposite 2-cycle $(x_{3}^h, x_{3}^t, x_{3}^h, x_{3}^t, x_{3}^h)$ is added in between the two blanks of $n_{i}$, $n_{j}$ and $n_{i'}$, $n_{j'}$. We just show the case that $i+1=j<i'=j'-1$, since all other cases are similar. 
$\mathbb{P}'=(0^h, \dots, n_{i}, x_{3}^h, x_{3}^t, n_{i'},$ $\dots, n_{j}, x_{3}^t, x_{3}^h, n_{j'}, \dots, n_{k},0^h)$ is also a path.

Therefore, in $ACG(G,X)$, all the blanks and blue edges form a path.
We can obtain $\tau$ along this path, if the path goes through a blue edge from $v_{a}^h$ to $v_{a}^t$, then it has a ``+''sign in $\tau$; and if the path goes through a blue edge from $v_{a}^t$ to $v_{a}^h$, then it has a ``-'' sign in $\tau$. Since the blanks represents adjacencies in both $\pi$ and $\tau$, then 
$\mathcal{A}[\pi]=\mathcal{A}[\tau]$. To make $\pi$ and $\tau$, we insert distinct genes between every adjacency.

Finally, we complete the proof by showing that $G$ has a Steiner set of size $k$ if and only if $\pi$ can be transformed into $\tau$ by $2|X|+2k+1$ symmetric reversals.

($\Rightarrow$) Assume that $G$ has a Steiner set $S$ of size $k$, which implies that the $G(X\cup S)$ forms a tree, 
thus, in the intersection graph $IG(\pi, \tau)$, the repeats corresponds to $X\cup S$ are in a single connected component, which contains a black vertex $x_{3}$. From our construction, there are also two odd vertices corresponding to each vertex of $X$, and one even vertex corresponding to each vertex of $S$, thus $\pi$ can be transformed into $\tau$ by $2|X|+2k+1$ symmetric reversals.   

($\Leftarrow$) Assume that $\pi$ can be transformed into $\tau$ by $2|X|+2k+1$ symmetric reversals. 
Since there are $2|X|+1$ odd vertices in the intersection graph $IG(\pi, \tau)$, but only one black vertices, 
then the solution must perform at least $2|X|+1$ symmetric reversals on these odd vertices, and at most $2k$ symmetric reversals on $k$ even vertices.
A symmetric reversal on a vertex is applicable if and only if the vertex is in a connected component that contains black vertices. Thus, the $2|X|+1$ odd vertices and $k$ even vertices are in a connected component, which implies that $G$
has a Steiner set of size $k$.
\end{proof}

We give an example of the above reduction in Figure~\ref{MSTandACG}.

\begin{figure}[H]
    \centering
    \includegraphics[width=1.0\textwidth]{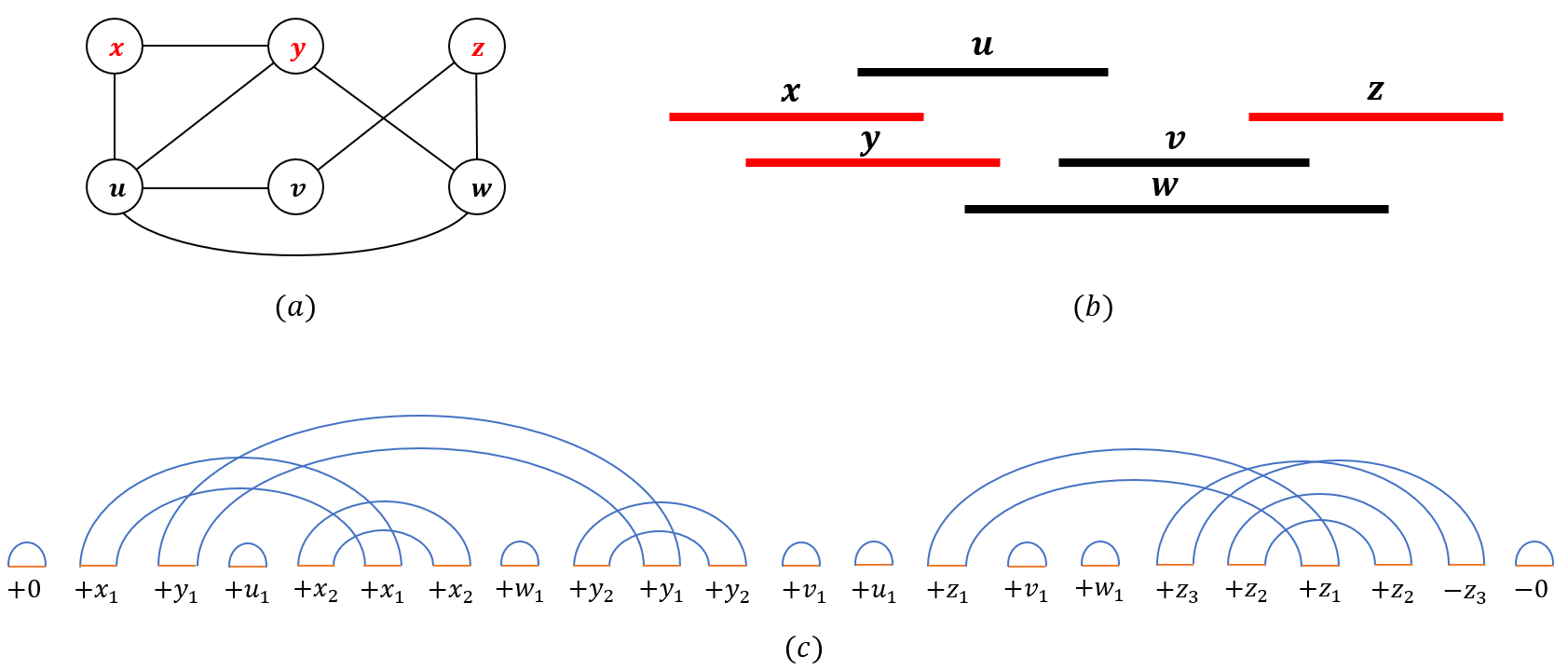}
    \caption{$(a)$. A circle graph $G$, where $X=\{x, y, z\}$ is the set of terminals. $(b)$. The corresponding intervals of $G$. $(c)$. The constructed alternative-cycle graph $ACG(G,X)$, and $\pi=[+0$, $+x_1,+y_1, +u_1,+x_2,+x_1,+x_2,+w_1$, $+y_2,+y_1,+y_2,+v_1,+u_1,+z_1,+v_1,+w_1,+z_3,+z_2,+z_1,+z_2,-z_3,-0]$, $\tau=[+0$, $+x_1,+x_2,+x_1,+y_1,+y_2,+y_1,+u_1,+x_2,+w_1,+y_2,+v_1,+u_1,+z_1,+z_2,+z_1,+v_1,\\+w_1,+z_3,-z_2,-z_3,-0]$.}
    \label{MSTandACG}
\end{figure}

\section{Concluding Remarks}
This paper investigates a new model of genome rearrangements named sorting by symmetric reversals. This model is based on recent new findings from genome comparison.  The decision problem, which asks whether a chromosome can be transformed into another by symmetric reversals, is polynomial solvable.
But the optimization problem, which pursues the minimum number of symmetric reversals during the transformation between two chromosomes, is NP-hard.
It is interesting to design some approximation algorithms for the optimization problem.
Perhaps, polynomial time algorithms to solve the optimization problem for more realistic special cases are also interesting.

\section*{Acknowledgments}
This research is supported by NSF of China under grant 61872427 and 61732009.

\end{document}